\documentclass[column,nofootinbib, superscriptaddress,10pt,aps,prx]{revtex4-2}
\usepackage[dvips]{graphicx} 
\usepackage{amsfonts}
\usepackage{csquotes} 
\usepackage{amssymb}
\usepackage{amscd}
\usepackage{amsmath}    
\usepackage{amsthm}
\usepackage{bm} 
\usepackage{bbm}
\usepackage{booktabs}
\usepackage{enumerate}
\usepackage{enumitem}
\usepackage{epsfig}
\usepackage{subfigure}
\usepackage{subfloat}
\usepackage{xcolor}	
\usepackage{physics}
\usepackage[most]{tcolorbox}
\usepackage{tikz}
\usetikzlibrary{quantikz}
\usepackage{tabularx}
\usepackage{float}
\usepackage{appendix}
\usepackage{makecell}
\usepackage{algpseudocode}  
\usepackage{dsfont}
\usepackage[colorlinks, linkcolor=blue, anchorcolor=blue, citecolor=blue]{hyperref}
\usepackage{MnSymbol}
\usepackage[linesnumbered,ruled,vlined]{algorithm2e}
\usepackage{natbib}
\usepackage{times}

\usepackage{tikz}
\usepackage{tikzpeople}
\usepackage{pgfplots}
\usetikzlibrary{quantikz}
\usetikzlibrary{arrows}
\usetikzlibrary{shapes,fadings,snakes}
\usetikzlibrary{decorations.pathmorphing,patterns}
\usetikzlibrary{calc}
\usetikzlibrary{positioning}
\definecolor{amber}{rgb}{1.0, 0.75, 0.0}

\newtheorem{theorem}{Theorem}
\newtheorem{fact}{Fact}
\newtheorem{lemma}{Lemma}
\newtheorem{corollary}{Corollary}
\newtheorem{claim}{Claim}

\newtheorem{definition}{Definition}

\newtheoremstyle{italichead_romanbody}
{\topsep}     
{\topsep}     
{\normalfont} 
{}            
{\itshape}    
{.}           
{0.5em}       
{}            

\theoremstyle{italichead_romanbody}

\makeatletter
\newcommand{\printappendixtoc}{%
	\section*{Contents}    
	\@starttoc{atoc}       
}

\newcommand{\l@atocsection}{\@dottedtocline{1}{0em}{2.3em}}
\newcommand{\l@atocsubsection}{\@dottedtocline{2}{3.5em}{2.3em}}
\makeatother

\makeatletter

\newcommand{\suppnote}[1]{%
	\renewcommand{\thesection}{\arabic{section}}%
	\refstepcounter{section}%
	\section*{Supplementary Note \arabic{section} -- #1}%
}

\makeatother

\graphicspath{{./figure/}}

\newcommand{\bE}{\mathbb{E}}

\newcommand{\cU}{\mathcal{U}}

\newcommand{\cA}{\mathcal{A}}
\newcommand{\cB}{\mathcal{B}}

\newcommand{\cD}{\mathcal{D}}

\newcommand{\cE}{\mathcal{E}}
\newcommand{\cO}{\mathcal{O}}

\newcommand{\Haar}{\mathrm{Haar}}

\newcommand{\poly}{\mathrm{poly}}
\newcommand{\negl}{\mathrm{negl}}

\renewcommand{\norm}[1]{\left\lVert#1\right\rVert}
\newcommand{\TV}[2]{D\left(#1, #2\right)}
\newcommand{\LWE}{\mathrm{LWE}}
\newcommand{\DLWE}{\mathrm{DLWE}}

\newcommand{\bmz}{{\bm{z}}}
\newcommand{\bmx}{{\bm{x}}}

\newcommand{\bms}{{\bm{s}}}

\newcommand{\ot}{\otimes}
\newcommand{\EE}[1]{\operatornamewithlimits{\mathbb{E}}_{#1}}

\definecolor{dred}{rgb}{.8,0.2,.2}
\definecolor{DRED}{rgb}{.8,0.2,.2}

\newtcolorbox[auto counter]{mybox}[2][]{
	enhanced,
	breakable,
	colback=blue!5!white,
	colframe=blue!75!black,
	fonttitle=\bfseries,
	title=Box \thetcbcounter: #2,#1
}

\begin{document}
	
	\title{No Cloning of Quantum Ensembles}	
	
	\author{Zhenyu Du}
	\thanks{These authors contributed equally to this work.}
	\affiliation{Center for Quantum Information, Institute for Interdisciplinary Information Sciences, Tsinghua University, Beijing 100084, China}
	
	\author{Siyuan Cheng}
	\thanks{These authors contributed equally to this work.}
	\affiliation{Center for Quantum Information, Institute for Interdisciplinary Information Sciences, Tsinghua University, Beijing 100084, China}
	
	\author{Qi Zhao}
	\email{zhaoqi@cs.hku.hk}
	\affiliation{QICI Quantum Information and Computation Initiative, School of Computing and Data Science, The University of Hong Kong, Pokfulam Road, Hong Kong}
	
	\author{Xiongfeng Ma}
	\email{xma@tsinghua.edu.cn}
	\affiliation{Center for Quantum Information, Institute for Interdisciplinary Information Sciences, Tsinghua University, Beijing 100084, China}
	
	\author{Xiao Yuan}
	\email{xiaoyuan@pku.edu.cn}
	\affiliation{Center on Frontiers of Computing Studies, Peking University, Beijing 100871, China}
	\affiliation{School of Computer Science, Peking University, Beijing 100871, China}
	
	\begin{abstract}
		Modern quantum physics now enables control of quantum systems at the level of individual trajectories, opening a new frontier that links quantum information theory, quantum many-body physics, and quantum thermodynamics, and uncovers novel non-equilibrium phenomena such as deep thermalization and measurement-induced entanglement.
		However, a central challenge remains: their characterization relies on measuring nonlinear properties of individual quantum states, a task tantamount to fine-grained cloning of a quantum ensemble.
		Here, the fundamental laws governing the cloning of quantum ensembles are investigated.
		First, a general no-cloning theorem for arbitrary ensembles is established from an information-theoretic perspective, even assuming multiple copies of the ensemble’s purification.
		It is then shown that this barrier can be unexpectedly circumvented for physical ensembles generated by finite-time evolutions.
		Nevertheless, these tasks are proven to remain computationally intractable, even when the full circuit description of state preparation is known.
		This stands in sharp contrast to the conventional no-cloning theorem, which relies on the state being unknown.
		Together, these results establish new fundamental principles of quantum mechanics, reveal intrinsic trade-offs among sample complexity, computational complexity, and quantum measurements, and highlight the necessity of problem-specific strategies for probing measurement-induced quantum phenomena.
	\end{abstract}
	
	\maketitle
	
	\section{Introduction}
	
	The notion of quantum ensemble—a probabilistic mixture of quantum states---plays central roles in both statistical physics and quantum information science~\cite{Landau1980StatisticalPhysics, Holevo1998Capacity}. 
	In conventional experiments, the inability to track microscopic information, together with unavoidable interactions with an inaccessible environment, necessitates describing the system by a mixed state, which represents the probabilistic mixture of a quantum ensemble. 
	Studies of the statistical properties of these ensembles have led to foundational concepts such as thermalization, with profound implications permeating many-body physics and information theory~\cite{Jaynes1957InformationTheory, Deutsch1991StatisticalClosedSystem, Srednicki1994ChaosThermalization, Rigol2008ThermalizationIsolated, Rahul2015ManyBodyLocalization}.
	
	Modern quantum experiments have established a new paradigm beyond this traditional setting. 
	State-of-the-art quantum platforms now allow full interrogation of many-body quantum systems, featuring high repetition rate, controllable unitary evolutions, and projective measurements~\cite{Bluvstein2024ReconfigurableAtomArrays, Norcia2023MidcircuitMeausurementRearrangement, Lis2023MidcircuitOperations, Graham2023MidcircuitMeasurements}. 
	This enables a shift from studying only the ensemble average to accessing its purification---a pure state that produces the ensemble via partial measurement, and further allows for tracking individual quantum trajectories (post-selected quantum states) as measurement outcomes are recorded. 
	Rich physical phenomena have hence been discovered by studying nonlinear properties of individual post-selected states in the ensemble, which are far beyond what a mixed-state description alone can capture. 
	Prominent examples include deep thermalization~\cite{Mark2024MaximumEntropy, Cotler2023EmergentDesign, Ho2022EmergentChaotic, Ippoliti2022SolvableModelDesignTime, Ippoliti2023DynamicalPurification, Claeys2022dualunitary, Choi2023RandomStateBenchmarking, Bhore2023ConstrainedSystem,  Chang2024ChargeConserving, Liu2024DeepThermalizationGaussian, Mok2024OptimalConversion, Cheng2025EmergentQEC, Chakraborty2025ComputationalDeepThermalization, Shrotriya2025NonlocalityDeepThermalization, Yu2025MixedStateDeepThermalization, Liu2025CoherenceDeepThermalization, Zhang2025HolographicDeepThermalization, Bejan2025MatchGateDeepThermalization}, measurement-induced entanglement~\cite{Li2018ZenoEffect, Li2019MIPT, Skinner2019MIPT, Cao2019MonitoringFermion, Chan2019UnitaryProjective, Bao2020TheoryMIPT, Choi2020QEC_MIPT, Jian2020MeasurementInducedCriticality, Gullans2020DynamicPurification, Zabalo2020CriticalProperties, Li2021ConformalInvariance, Ippoliti2021MeasurementOnlyDynamics, Alberton2021FermionTransition, Fan2021SelfOrganizedErrorCorrection, Agrawal2022MonitoredQuantumCircuits, Noel2022MeasurementInducedPhaseTrappedIon, Koh2023MeasurementInducedEntanglement, Hoke2023MeasurementInducedEntanglement, Kawabata2023MIPT_NonHermitian, Fuji2020CriticalityContinuousMonitoring, Ippoliti2024LearnabilityTransition, Kamakari2025CrossEntropySuperconducting}, magic~\cite{Niroula2024MagicPhaseTransition, Zhang2024MagicConcentration, Loio2025quantumstatedesignsmagic, Du2025CertifyingLocalizable}, and complexity~\cite{Du2025CertifyingLocalizable, Du2025SpacetimeComplexity, McGinley2025MeasurementInduced2D, Suzuki2025quantumcomplexity}.
	Collectively, the new paradigm of studying quantum ensembles is deepening our understanding of fundamental mechanisms such as thermalization~\cite{Mark2024MaximumEntropy} and quantum error correction~\cite{Bravyi2005Universal, Gullans2020DynamicPurification, Choi2020QEC_MIPT, Fan2021SelfOrganizedErrorCorrection, Cheng2025EmergentQEC}.
	
	Despite these advances, this emerging paradigm faces a central paradox: nonlinear properties are essential to the phenomena of interest, yet they are often difficult to access experimentally.
	Their observation typically requires multiple copies of a post-selected quantum state, but the intrinsically stochastic nature of quantum measurements makes the production of even two identical copies from an ensemble highly challenging.
	As a consequence, current experimental studies are confined to highly specialized regimes~\cite{Noel2022MeasurementInducedPhaseTrappedIon, Koh2023MeasurementInducedEntanglement, Hoke2023MeasurementInducedEntanglement, Choi2023RandomStateBenchmarking, Niroula2024MagicPhaseTransition, Kamakari2025CrossEntropySuperconducting}.
	This tension raises a fundamental question: do intrinsic constraints limit such ensemble-processing tasks? Addressing this question is essential both for understanding the true nature of these emerging phenomena and for clarifying the borderline of quantum information processing.
	
	In this work, the fundamental laws governing ensemble cloning are examined from both information-theoretic and computational perspectives. We first demonstrate that no sample-efficient algorithm exists for cloning states within an ensemble or probing their nonlinear properties, even when the algorithm is granted access to multiple copies of the ensemble's purification (see Fig.~\ref{fig:settings}). This result goes beyond the conventional no-cloning theorem~\cite{Wootters1982NoCloning}: while the traditional hardness arises from the indistinguishability of non-orthogonal quantum states, the hardness we identify further stems from a fundamental tradeoff between random measurements and sample complexity. 
	For instance, conventional cloning via state tomography requires only constant sample complexity for small systems.
	In contrast, the complexity of our cloning task can be exponentially large even for a single-qubit ensemble, provided the trajectory space is sufficiently large.
	This rigorously establishes that the nonlinear properties of quantum ensembles are, in general, \emph{information-theoretically} difficult to access.
	
	\begin{figure*}[!htbp]
		\centering
		\includegraphics[width=.8\linewidth]{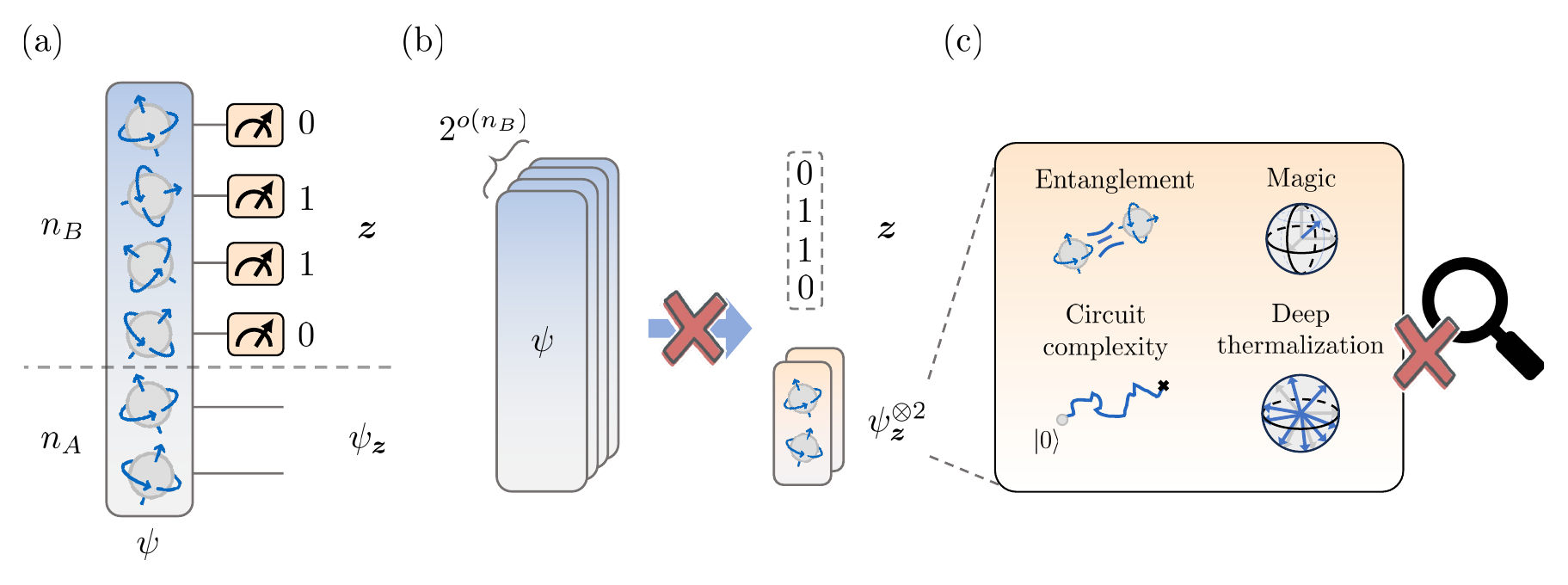}
		\caption{Settings of cloning a quantum ensemble and estimating its nonlinear properties. 
			(a) We consider an operational setting in which repeated copies of the pure state $\ket{\psi}_{AB} = \sum_{\bmz} \sqrt{p(\bmz)} \ket{\psi_{\bmz}}_A \otimes \ket{\bmz}_B$ are accessible. This state is a purification of the quantum ensemble $\{(p(\bmz),\ket{\psi_{\bmz}})\}_{\bmz}$ on system $A$.  Local projective measurements on the environment $B$ in the computational basis then generate the ensemble on $A$. 
			(b) No-cloning of quantum ensembles. We prove that cloning states in the ensemble is hard on average, even with sub-exponential access to its purification. 
			(c) Hardness of estimating nonlinear properties. We prove this task is hard both information-theoretically and computationally in the worst case. This hardness precludes efficient algorithms for observing a wide range of measurement-induced quantum phenomena.
		}
		\label{fig:settings}
	\end{figure*}
	
	While the results above exclude universal algorithms, quantum ensembles appearing in the lab generally have an efficient preparation circuit~\cite{Poulin2011ConvenientIllusion}. 
	We show that this prior knowledge indeed enables a sample-efficient algorithm for ensemble cloning, effectively bypassing the information-theoretic barrier. 
	Yet, the algorithm still requires exponential computational time for classical post-processing.
	This immediately prompts the question: is it computationally hard to clone ensembles generated from such physical procedures?    
	This leads to another central result of our work: we show that the ensemble cloning and estimation tasks remain \emph{computationally} hard, even with prior knowledge that the purified state is prepared by a polynomial-size quantum circuit, or even given a complete classical description of that circuit. 
	This finding stands in stark contrast to the conventional no-cloning theorem, which relies strictly on the premise that the quantum state is unknown. 
	Our results hence reveal a new computational barrier in quantum information processing and many-body physics by establishing unclonability in the presence of complete structural knowledge.
	
	Taken together, these results uncover fundamental laws governing the ensemble-cloning task that arise from mechanisms extending beyond the traditional no-cloning principle. These include: (1) information-theoretic hardness, whereby sub-exponential access to an unknown purified state is insufficient to extract nonlinear ensemble properties; and (2) computational hardness, whereby even when the purified state is known and efficiently preparable, its nonlinear ensemble properties remain protected by a computationally hard problem, precluding their extraction by any efficient algorithm.
	
	Our work has immediate implications for measurement-induced phenomena, such as deep thermalization and measurement-induced entanglement.
	While these phenomena are well-studied in various solvable models, our results show that it is intrinsically hard to obtain their experimental signatures in the worst case. 
	This hardness holds broadly: (1) for state-agnostic algorithms, such as those used for probing measurement-induced phenomena in an arbitrary unknown quantum system, and (2) for state-aware algorithms, used for observing or verifying those phenomena even when assisted with full circuit-level knowledge. 
	Our findings thus prove that, for the first time, observing measurement-induced phenomena beyond classically predictable regimes is generally hard. This motivates further research into devising efficient strategies and identifying physically relevant regimes where efficient characterization remains feasible despite generic hardness.
	
	Moreover, our work identifies a constraint on recent proposals that utilize measurement-based randomness to derandomize quantum
	circuits~\cite{Turner2016Derandomizing, Cotler2023EmergentDesign}. While such methods have practical benefits---as reconfiguring quantum circuits is often more demanding than repeating experiments in a fixed setting---we reveal a major drawback: a single random circuit cannot be used twice.
	Taken together, our results delineate a new boundary on the power of quantum computers.

	\section{Results}
	\subsection{Quantum Ensembles}
	A quantum ensemble $\mathcal{E} = \{(p(\bmz), \ket{\psi_{\bmz}})\}_{\bmz}$ is a probability distribution over quantum states. 
	In this work, we focus on ensembles of pure states, while our no-go results readily extend to ensembles of mixed states. 
	In principle, any ensemble $\mathcal{E}$ can be purified by a pure state $\ket{\psi}_{AB} = \sum_{\bmz} \sqrt{p(\bmz)} \ket{\psi_{\bmz}}_A \otimes \ket{\bmz}_B$ on a larger $n$-qubit system $AB$ with system $A$ and environment $B$. The sizes of  
	subsystems $A$ and $B$ are $n_A$ and $n_B$, respectively, with the whole system size $n = n_A + n_B$.
	Local projective measurements on $B$ in computational basis then generate the ensemble on system $A$ (Fig.~\ref{fig:settings}a). Therefore, we refer to $\mathcal{E}$ as the projected ensemble of $\psi_{AB}$.
	
	We consider the scenario where copies of the pure state $\psi_{AB}$ are accessible, which aligns with modern quantum platforms that feature high-fidelity control, read-out, and high-rate repetition. 
	This setting is widely adopted in the study of quantum simulation, such as deep thermalization~\cite{Ho2022EmergentChaotic, Ippoliti2022SolvableModelDesignTime, Claeys2022dualunitary, Cotler2023EmergentDesign} and measurement-induced phase transitions~\cite{Li2018ZenoEffect, Li2019MIPT, Skinner2019MIPT}. 
	Beyond its foundational role, the purification also finds useful applications ranging from derandomizing quantum circuits~\cite{Turner2016Derandomizing} to device benchmarking~\cite{Choi2023RandomStateBenchmarking} and characterization~\cite{Tran2023MeasuringArbitrary, Mok2024OptimalConversion}.
	
	The task of ensemble cloning consists of using multiple copies of $\psi_{AB}$ to produce the ensemble $\{(p(\bmz), \psi_{\bmz}^{\otimes 2})\}$, where the measurement outcomes $\bmz$ are sampled from the same distribution as in the projected ensemble $\mathcal{E}$. Mathematically, this corresponds to preparing the classical–quantum state
	\begin{equation}
		\sigma^{(2)}_{\psi} \coloneqq \sum_{\bm z} p(\bm z)\, \psi_{\bm z}^{\otimes 2} \otimes \ketbra{\bm z},
	\end{equation}
	where the first register holds the states and the second stores the measurement outcome. 
	This should be distinguished from generating two independent samples from the projected ensemble, which corresponds to resampling rather than cloning a single sampled projected state and is straightforward given two copies of $\psi_{AB}$.
	We require the output ensemble to follow the same probability distribution $p(\bm z)$, mirroring the requirements for averaging over the ensemble. This requirement can be relaxed, as we discuss later.
	The cloning task is not only a natural primitive for probing higher-order moments of the ensembles, such as subsystem purity or the variance of an observable, but also operationally significant for protocols that need to reuse a randomly sampled state twice or more, such as when the ensemble forms pseudorandom ensembles~\cite{Joseph2003PseudoRandomUnitary, Turner2016Derandomizing, Cotler2023EmergentDesign}. 
	
	We also consider the task of estimating nonlinear properties of an ensemble. 
	We define the second moment state of the projected ensemble as $\rho_{\psi}^{(2)} \coloneqq \sum_{\bm z} p(\bm z)\, \psi_{\bm z}^{\otimes 2}$. 
	A second-order nonlinear quantity averaged over $\mathcal{E}$ is then given by
	\begin{equation}
		o_{\psi}^{(2)} \coloneqq \tr[O\,\rho_{\psi}^{(2)}]
	\end{equation}
	for a given observable $O$ acting on $\mathcal{H}_A^{\otimes 2}$. 
	The above definitions straightforwardly extend to higher-order $\rho_\psi^{(k)}$ and $o_{\psi}^{(k)}$ for $k\ge 2$. While here we focus on the simplest nontrivial case, $k=2$, the no-cloning theorems apply straightforwardly to general $k\ge 2$.
	Estimating nonlinear properties is expected to be easier than directly cloning the state, as it only requires outputting a scalar value rather than preparing a high-dimensional state. 
	Indeed, the hardness of estimation implies the hardness of cloning, but not vice versa.
	
	The estimation task provides the framework for observing measurement-induced phenomena.
	For example, $o_{\psi}^{(2)}$ captures the ensemble-averaged subsystem purity, a key signature of measurement-induced entanglement. 
	To see this, partition $A$ into $A = L \cup R$. The purity of a given projected state $\ket{\psi_z}$ is $\tr(\rho_{z,L}^2)$, where $\rho_{z,L} \coloneqq \tr_R[\psi_z]$ is the reduced density matrix on $L$. 
	The ensemble-averaged purity, $\mathbb{E}_{z}[\tr(\rho_{z,L}^2)]$, is then precisely $o_{\psi}^{(2)}$ for the swap operator $O_E \coloneqq \mathbb{S}_{L_1L_2} \otimes I_{R_1R_2}$.
	Similarly, measurement-induced magic can be quantified by $o_{\psi}^{(4)}$ using the observable $O_M \coloneqq 4^{-n_A}\sum_{P} P^{\otimes 4}$ acting on $\mathcal{H}_A^{\otimes 4}$, where the summation runs over all $n_A$-qubit Pauli operators~\cite{Oliviero2022MeasuringMagic}.
	
	Moreover, deep thermalization is defined by the convergence of the higher-order moment states $\rho_\psi^{(k)}$ to their thermal counterparts $\tau^{(k)} \coloneqq \bE_{\phi \sim \mu} \phi^{\otimes k}$,
	where the ensemble $\mu$ is described by the Haar-random distribution on system $A$ at the infinite temperature, or by the Scrooge ensemble at finite temperature~\cite{Mark2024MaximumEntropy}. 
	Certifying this convergence, $\rho_\psi^{(k)} \approx \tau^{(k)}$, is equivalent to showing $o_{\psi}^{(k)} \approx \tr[O\tau^{(k)}]$ for all normalized $O$. Therefore, the hardness of estimating $o_{\psi}^{(k)}$ directly implies a fundamental barrier to verifying deep thermalization.
	
	Having outlined the relevant settings, we now proceed to establish the fundamental laws of the ensemble cloning tasks from both information-theoretic and computational perspectives.
	We first identify the sample complexity barriers that arise in the absence of prior knowledge. Next, we show how to overcome these limitations by leveraging prior information about the state preparation. 
	Finally, we establish that despite this gain in sample efficiency, these tasks remain computationally intractable even when a full classical description of the circuit is provided.
	
	\subsection{Information-theoretical hardness}

	First, we study the sample complexity required to achieve these tasks, given access to copies of the purification of an \emph{unknown} ensemble. This setting is not only standard in quantum information processing tasks, but also highly relevant in practice, as characterizing an unknown state generally requires exponential resources~\cite{Haah2016OptimalTomo}.
	
	We note that previous works have already highlighted some limitations in this setting.
	The standard no-cloning theorem implies the hardness of cloning an individual state given access to a single copy~\cite{Wootters1982NoCloning}.
	More recently, it has been proven that distinguishing an unknown ensemble $\{p(\bmz), \psi_{\bmz}\}$ from its probabilistic mixture $\{p(\bmz), \rho\}$, where $\rho = \sum_\bmz p(\bmz)\psi_{\bmz}$, is generally hard when one can only sample directly from the ensemble~\cite{McGinley2024PostselectionFreeLearning}. Yet, a fully controlled quantum device may access multiple copies of the purification and perform coherent joint operations across them. A priori, it is plausible that such highly entangled operations could extract nonlinear information that is inaccessible to direct-sampling strategies~\cite{Chen2022Memory, Huang2022QuantumAdvantage}.
	
	Our result goes beyond these known limitations. We show that it is hard to clone states in an ensemble or estimate their nonlinear properties, even when one is given access to sub-exponential many copies of the ensemble's purification and is allowed to perform arbitrary entangled operations across those copies. 
	Therefore, our result rules out efficient quantum algorithms with a substantially stronger access model~\cite{Chen2022Memory, Liu2025ExponentialSeparationPurification}. 
	This setting aligns with the capabilities of advanced quantum devices. The no-go result is formulated as follows:
	
	\begin{theorem}[No-cloning of quantum ensembles, informal]\label{thm:summary}
		Any quantum algorithm that takes copies of an unknown purification state $\ket{\psi}_{AB}$ of an ensemble $\mathcal{E}$ as input and aims to clone individual states (preparing $\sigma_{\psi}^{(2)}$) or estimate nonlinear properties (outputting $o_{\psi}^{(2)}$ for certain observables $O$ satisfying $\norm{O}_{\infty} \le 1$) with a small constant error requires $2^{\Omega(n_B)}$ sample complexity.
		
		This lower bound holds in the average case (over Haar-random $\psi_{AB}$) for the cloning task, and in the worst case (over $\psi$ and $O$) for the estimation task.
	\end{theorem}
	
	Our proofs utilize the symmetry of the Haar-random unitary and permutation groups, as well as the distinguishability between Haar-random states and random subset states~\cite{Tudor2023SubsetStates, Fernando2024SubsetStates}. These are detailed in Supplementary Notes~2 and~3, respectively.
	The average-case result therefore shows that cloning is hard for typical ensembles induced by Haar-random purifications, a natural distribution over ensembles in the studies of both traditional thermalization~\cite{Zyczkowski2001InducedMixedState, Popescu2006EntanglementStatistical} and measurement-induced settings such as deep thermalization~\cite{Cotler2023EmergentDesign}.
	We remark that average-case hardness does not apply to the estimation of nonlinear properties. For a typical Haar-random state, the nonlinear value $o_{\psi}^{(2)}$ for a specific observable $O$ concentrates around the known thermal value $\tr[O\tau^{(2)}]$ because typical states are already deeply thermalized~\cite{Cotler2023EmergentDesign}, making the estimation task trivial on average. 
	Nonetheless, many physically relevant ensembles, such as those arising in monitored quantum circuits, are not purified by typical random states. 
	Consequently, the established worst-case hardness rules out universal algorithms for the ensemble estimation task.
	
	The hardness of cloning leads to several profound implications. (1) In scenarios such as measurement-based circuits or deep thermalization, where purification states find applications in derandomizing random circuits~\cite{Turner2016Derandomizing, Cotler2023EmergentDesign, Tran2023MeasuringArbitrary, Mok2024OptimalConversion}, we show that a single random state from the projected ensemble cannot be reused. 
	Consequently, such measurement-induced randomness cannot be utilized in randomized measurement protocols that require the repeated application of the same random instance~\cite{Brydges2019Probing}. This highlights a drawback compared to using controllable unitary circuits.
	(2) Many current experiments probe ensemble properties via post-selection~\cite{Koh2023MeasurementInducedEntanglement, Hoke2023MeasurementInducedEntanglement, Niroula2024MagicPhaseTransition, Hou2025MachineLearningEffects}, a process that is essentially equivalent to cloning states. 
	Our results show that the difficulty of such a process is intrinsic.
	(3) It generalizes the traditional no-cloning theorem ($t=1$, $n_B=0$) by quantifying an intrinsic tradeoff between the number of measurements $n_B$ and the sample complexity $t$. 
	While the traditional theorem relies on the input state being available in only finite copies, our result reveals a novel measurement-induced mechanism that forbids cloning even when the algorithm has access to multiple copies of the purification.
	See Supplementary Note~2 for a detailed tradeoff between the approximation error, sample complexity, and measurements.
	
	The hardness of estimating nonlinear properties stems from a distinguishing game. While global states drawn from random subset ensembles and Haar-random ensembles are statistically hard to distinguish, their projected states are drastically different: projected states from the subset ensemble are simple computational-basis states, whereas those from the Haar-random ensemble are highly random and entangled.
	Hence, any algorithm capable of estimating $o_{\psi}^{(2)}$ of certain observables $O$ to constant additive accuracy would solve this hard distinguishing problem, implying it must use exponentially many copies.  
	This result establishes that no universal algorithm can estimate a two-copy observable efficiently as a linear observable on reduced density matrices. 
	While Theorem~\ref{thm:summary} establishes a worst-case result over the observable $O$, this hardness extends more broadly to properties that distinguish between computational-basis ensembles and projected ensembles of Haar-random states:
	\begin{corollary}[Hardness of observing measurement-induced quantum phenomena,  informal] \label{col:hardness_phenomena_informal}
		For any property that distinguishes ensembles of computational-basis states from projected ensembles of Haar-random states (such as entanglement, magic, circuit complexity, or forming a state design), estimating this property over an unknown measurement-induced ensemble requires a sample complexity exponential in the number of measured sites in the worst case. 
	\end{corollary}
	The formal statement of this corollary is provided in Supplementary Note~3. These results directly imply the worst-case hardness of observing measurement-induced entanglement, magic, complexity, and deep thermalization. Notably, while deep thermalization represents a stronger and more general notion of equilibrium~\cite{Mark2024MaximumEntropy}, this generality comes at a price: it lacks the efficient verifiability of conventional thermalization, which can be checked by measuring local observables and comparing them to thermal predictions. 
	For the infinite-temperature case, our result directly implies the hardness of verifying emergent state designs~\cite{Cotler2023EmergentDesign, McGinley2024PostselectionFreeLearning, Nakata2025ComputationalComplexityDesigns}.
	
	Our results hold even when the distribution of states we considered is not identical to that in the original ensemble, as long as it is concentrated on typical trajectories that are representative of the ensemble (i.e., product states and highly entangled states). 
	Since the hardness of estimation implies the hardness of cloning, this establishes the hardness of ensemble cloning even when one only needs to clone a typical state of the ensemble.
	
	\subsection{Sample-efficient algorithm with prior knowledge}
	While the above results rule out universal, state-agnostic algorithms, practical experiments generally possess some degree of prior knowledge regarding the ensemble. This prompts another interesting question: can we use prior knowledge to surpass these fundamental limitations? We now proceed to answer this question. 
	
	We focus on two scenarios that are common in practice but not so specific as to limit our scope to the classically simulable regime:
	\begin{enumerate}[label={$(\arabic*)$}]
		\item The algorithm knows that the ensemble has a polynomial-size preparation circuit, but does not know the specific circuit. 
		\item The algorithm is given a full classical description of the circuit used to prepare the ensemble. 
	\end{enumerate}

	Prior knowledge (1) is a reasonable assumption for many ensembles encountered in experiments, where the system is initialized in some fiducial states (e.g., computational-basis states) and then undergoes a finite-time evolution~\cite{Poulin2011ConvenientIllusion}. 
	Consequently, the purification states of such ensembles can be modeled by finite-sized circuits consisting of $G = \poly(n)$ two-qubit gates. 
	We demonstrate that this form of prior knowledge is already sufficient to bypass the information-theoretic barrier established in Theorem~\ref{thm:summary}.
	In particular, one can learn a classical description of a state approximating the higher moments $\sigma_\psi^{(k)}$ of the ensemble using a number of copies of the purification polynomial in the circuit size $G$ and the moment order $k$. Once this classical description is obtained, one can prepare the cloned ensemble state or compute its nonlinear quantities, although the required classical post-processing may still be computationally intractable.
	
	\begin{theorem}[Sample-efficient algorithm for cloning ensembles of bounded gate complexity]\label{thm:sample-efficient-learning}
		Given $0 < \varepsilon < 1$ and $N$ copies of an $n$-qubit input state $\ket{\psi}_{AB} = U\ket{0}^{\otimes n}_{AB}$, where the unitary $U$ consists of $G$ two-qubit gates, one can obtain a classical description of a state $\tilde{\sigma}$ such that $D(\tilde{\sigma}, \sigma_{\psi}^{(k)}) \le \varepsilon$ provided that $N = \widetilde{\Theta}\left(\frac{k^2G}{\varepsilon^2}\right)$. Here, $\widetilde{\Theta}(\cdot)$ hides logarithmic factors.
	\end{theorem}
	The proof is provided in Supplementary Note~4. 
	This theorem builds on the sample-efficient algorithm for learning quantum states of bounded gate complexity~\cite{Zhao2024BoundedGateComplexity}. In particular, one can learn the bounded-complexity purification up to trace distance $\varepsilon'$ using $\widetilde{\Theta}(G/\varepsilon'^2)$ copies. We then show that trace-distance error $\varepsilon'=\varepsilon/(2k+1)$ in the purification is sufficient to guarantee trace-distance error $\varepsilon$ for the $k$th-order moment of the ensemble. Combining these two ingredients establishes the theorem.

	These results indicate that quantum ensembles efficiently preparable in experiments can be sample-efficiently cloned, enabling the study of measurement-induced phenomena in practically relevant settings.
	In particular, the theorem guarantees ensemble cloning with only polynomial sample complexity, corresponding to a polynomial number of queries to the purification unitary.
	This perspective helps explain why measurement-induced phenomena can be observed using very few samples in current experiments~\cite{Gullans2020ScalableProbes, Li2023XEBbenchmarkMIPT, McGinley2024PostselectionFreeLearning, Garratt2024PostmeasurementPostselection, Kamakari2025CrossEntropySuperconducting, Du2025CertifyingLocalizable, Kim2025MachineLearningMIPT}.
	However, a central limitation of both the proposed protocol and existing experimental approaches is their reliance on classical post-processing that becomes intractable once the underlying circuits are no longer classically simulable.
	As a result, these methods may fail to scale to large system sizes without incurring exponential computational cost.
	This scalability bottleneck motivates a rigorous examination of the computational complexity barriers underlying ensemble-processing tasks.
	
	\subsection{Computational hardness with prior knowledge} 
	
	We now proceed to investigate the computational complexity of these tasks given prior knowledge (1) or the even stronger condition (2).
	It corresponds to the setting of current programmable quantum platforms, where the state preparation circuit is fully controllable and known.
	This is a much stronger form of prior knowledge, as the underlying ensemble is fully specified information-theoretically.
	As we will show, this prior knowledge is still insufficient to allow for computationally efficient algorithms for the ensemble processing tasks.
	
	Our hardness results rely on standard computational assumptions. The first part of our theorem assumes the existence of quantum-secure pseudo-random functions (QPRF) and quantum-secure pseudo-random permutations (QPRP), which are widely believed to exist~\cite{Zhandry2021QuantumRandomFunctions, Zhandry2025notequantumsecure}. The second part assumes the quantum hardness of the LWE problem, a widely-adopted assumption in post-quantum cryptography~\cite{regev2009lattices}. We summarize our result in the following theorem, leaving the complete proof in Supplementary Note~5.
	
	\begin{theorem}[No-go on estimating nonlinear properties with prior knowledge]\label{thm:no-go-prior}
		Let $0\leq\varepsilon< \tfrac{1}{10}$ and $0\leq\delta<\tfrac13$ be fixed constants. 
		Given a partition $A\mid B$ and an efficiently measurable observable $O$, suppose a quantum algorithm $\mathcal{A}$ estimates $o_{\psi}^{(2)}$ to within an additive error $\varepsilon$ with probability at least $1-\delta$.
		\begin{enumerate}[label={$(\arabic*)$}]
			\item Let $n_A \ge 1, n_B = \omega(\log n)$. If $\mathcal{A}$ is given copy access to a state $\ket{\psi}$ that is promised to be efficiently preparable, then $\mathcal{A}$ requires super-polynomial computational time, assuming the existence of QPRF and QPRP. 
			Here, the computational time includes the cost of accessing and processing the input copies, with each access counted as one unit of time.
			
			\item Let $n_A,n_B = \Omega(n^c)$ for some constant $0 < c \le 1$. If $\mathcal{A}$ takes as input the classical description of a polynomial-size circuit $C$ that prepares $\ket\psi = C\ket{0}^{\otimes n}$, then $\mathcal{A}$ requires super-polynomial computational time, assuming the LWE problem cannot be solved by any polynomial-time quantum algorithm.
		\end{enumerate}
	\end{theorem}
	\noindent Here, efficiently measurable means $O = U \left(\sum_{\bm z} f(\bm z) \ketbra{\bm z}\right) U^{\dagger}$ where $f$ is an efficiently computable function and $U$ is a polynomial-size quantum circuit. This restriction isolates the intrinsic difficulty of the estimation problem itself from the complexity of the observable.
	By further strengthening the assumption of Theorem~\ref{thm:no-go-prior} to pseudo-random primitives that remain secure against sub-exponential-time (i.e. $2^{o(n)}$-time) quantum adversaries—such as those derived from the LWE problem~\cite{regev2009lattices, Zhao2024BoundedGateComplexity}—the lower bound of computational time in (1) can be extended to $2^{\Omega(n)}$, and the system-size requirement in (2) can be relaxed to $n_A, n_B \ge \widetilde{\Theta}(\log n)$.
	
	This effectively rules out any efficient universal algorithm that takes as input a two-copy observable $O$ and outputs the value $o_{\psi}^{(2)}$, even when provided with the full circuit description of the state. 
	Consequently, ensemble cloning is forbidden by computational barriers even in the presence of full state knowledge. 
	This reveals a mechanism distinct from the traditional no-cloning theorem, which relies strictly on the input state being unknown.
	
	Notably, our proof technique can also be used to show that estimating ensemble-averaged entanglement is at least as hard as solving the LWE problem, thereby establishing a computational barrier for observing measurement-induced entanglement.
	\begin{corollary}[Hardness of observing measurement-induced entanglement with prior knowledge]  
		\label{corollary:entanglement}
		Let $n_A, n_B \ge \widetilde{\Theta}(\log n)$. Even when provided with a full circuit description of the purification state $\ket{\psi}_{AB}$ of an ensemble, estimating the average entanglement over this ensemble requires $n^{\omega(1)}$ computational time in the worst case, assuming the LWE problem cannot be solved by any sub-exponential-time quantum algorithm.
	\end{corollary}
	The detailed proof is provided in Supplementary Note~5. We anticipate that similar computational barriers hold for other nonlinear properties, including magic, circuit complexity, and deep thermalization.
	This implies that ``truly quantum'' signatures of measurement-induced phenomena—those beyond the classically-predictable regime—are generally computationally hard to observe in the state-aware scenario. 
	This explains why current sample-efficient schemes for estimating nonlinear properties of quantum ensembles highly rely on the classical simulability of the preparation circuits~\cite{Gullans2020ScalableProbes, Li2023XEBbenchmarkMIPT, McGinley2024PostselectionFreeLearning, Garratt2024PostmeasurementPostselection}. 
	
	Finally, our results demonstrate that estimating the simplest nonlinear properties for $k=2$ already gives computational capabilities beyond standard polynomial-time quantum algorithms.
	An intriguing avenue for future research is to determine whether estimating higher-order properties for $k > 2$ yields further computational power.
	Specifically, one might investigate whether access to the $k$-th moment for $k = \poly(n)$ provides computational power equivalent to post-selection (i.e., the complexity class \textsf{PostBQP}). 
	Establishing such a relationship would delineate a clear computational hierarchy based on the order of the ensemble moments.
	
	\section{Discussion}
	We have established fundamental information-theoretic and computational laws to cloning quantum ensembles and estimating their properties. Our work highlights that many novel quantum phenomena defined over ensembles are provably hard to observe efficiently in the worst case. This underscores that experimental feasibility should be an important consideration when exploring new physical concepts~\cite{Feng2025HardnessStrongToWeak, Schuster2025HardnessPhases}.

	While our general no-go theorems rule out efficient, universal algorithms, exceptions exist when the dynamics are carefully structured. 
	For instance, strategies based on spacetime duality can map monitored dynamics to unitary dual circuits, enabling efficient purity estimation of post-measurement states~\cite{Ippoliti2021PostselectionFreeSpacetimeDuality}.
	Nonetheless, these structural restrictions often eliminate critical behaviors and render the dynamics analytically solvable~\cite{Lu2021SpacetimeDualLocalization, Claeys2022dualunitary, Claeys2022ExactDynamics, Ippoliti2022SteadyStateSpacetimeDual}. Consequently, current efficient protocols are often constrained to scenarios where the dynamics are classically tractable~\cite{Noel2022MeasurementInducedPhaseTrappedIon, Koh2023MeasurementInducedEntanglement, Hoke2023MeasurementInducedEntanglement, Choi2023RandomStateBenchmarking, Niroula2024MagicPhaseTransition, Kamakari2025CrossEntropySuperconducting}.
	
	Meanwhile, other quantum ensembles that are not directly induced by measurements have also recently attracted attention. A prominent example is the temporal ensemble of a Hamiltonian, $\left\{(1/\tau, e^{-iHt} \ket{\psi_0})\right\}_{t = -\tau/2}^{\tau / 2}$, which plays a central role in studies of Hilbert-space ergodicity~\cite{Mark2024MaximumEntropy, Saul2024ErgodicityDesign}. 
	Unlike the measurement-induced projected ensemble, each state in this temporal ensemble can be efficiently prepared via Hamiltonian simulation. This enables efficient estimation of the nonlinear quantities considered here and potentially allows for the observation of new phenomena~\cite{Shaw2025ExperimentalErgodicity}. 
	Furthermore, exploring other physically motivated ensembles, particularly those arising from thermodynamics, and investigating how the hardness results established here relate to fundamental thermodynamic principles presents an intriguing avenue for future research.
	
	Finally, our results point to a central dilemma for nonlinear phenomena in quantum ensembles: either they are hard to access even for quantum computers, or they are accessible only in regimes where efficient classical methods are already available. This highlights the central challenge in the field: can we identify novel phenomena whose observation genuinely requires quantum processors and goes beyond classical computation? Answering this question is key to realizing meaningful quantum speedups in scientific discovery~\cite{Lloyd1996UniversalSimulator, Kim2023QuantumUtility, King2025BeyondClassicalComputation, Abanin2025OTOC}, and we hope our work will further motivate both impossibility and possibility results.
	
	\section*{Acknowledgments}
	We thank Wen Wei Ho for helpful discussions that inspired the idea behind this project and for his comments on the manuscript.
	X.Y. is supported by Beijing Natural Science Foundation Z250004, the National Natural Science Foundation of China NSAF (Grant No.~U2330201) and  Grant (No.~12361161602),  the Quantum Science and Technology-National Science and Technology Major Project (2023ZD0300200),   and Beijing Science and Technology Planning Project (Grant No.~Z25110100810000). 
	Z.D., S.C. and X.M. acknowledge the support from the National Natural Science Foundation of China Grants No.~12174216,  the Innovation Program for Quantum Science and Technology Grant No.~2021ZD0300804,  No.~2021ZD0300702, the CCF-QuantumCtek Superconducting Quantum Computing Special Cooperation Program (Grant No.~CCF-QC2025005), and the Turing AI Institute of Nanjing.
	Q.Z. acknowledges funding from National Natural Science Foundation of China (NSFC) via Project No. 12347104 and No. 12305030, Guangdong Basic and Applied Basic Research Foundation via Project 2023A1515012185, Hong Kong Research Grant Council (RGC) via No. 27300823, N\_HKU718/23, and R6010-23, Guangdong Provincial Quantum Science Strategic Initiative No. GDZX2303007. 
	
	\bibliography{ref}

	\clearpage
	\onecolumngrid
	\setcounter{section}{0}

	\begin{center}
		{\large \textbf{Supplementary Information}}
	\end{center}
	\section*{Contents} 
	\makeatletter
	\let\oldaddcontentsline\addcontentsline
	\renewcommand{\addcontentsline}[3]{%
		\def\target{#1}%
		\def\toclabel{toc}%
		\ifx\target\toclabel
		\oldaddcontentsline{atoc}{atoc#2}{#3}%
		\else
		\oldaddcontentsline{#1}{#2}{#3}%
		\fi
	}
	
	\@starttoc{atoc}
	\makeatother
	
	\vspace{1cm}


	\suppnote{Preliminaries on random unitaries and states}
	In this section, we introduce the preliminaries used in this work. We denote the $d$-dimensional Hilbert space as $\mathcal{H}_d$.
	The Haar-random distribution of unitaries and states on $n$-qubits are both denoted by $\mathrm{Haar}(n)$, which can be specified from the context. 
	The permutation group of $k$ elements is denoted $S_k$.
	For a permutation $\pi \in S_k$, the corresponding permutation operator $V_\pi$ on $\mathcal{H}_d^{\otimes k}$ is given by 
	\begin{equation}
		V_\pi(\ket{\psi_1}\cdots\ket{\psi_k}) = \ket{\psi_{\pi^{-1}(1)}} \cdots \ket{\psi_{\pi^{-1}(k)}},
	\end{equation}
	where $\ket{\psi_1}, \cdots, \ket{\psi_k} \in \mathcal{H}_d$.
	Define the operator $P_{\mathrm{sym}}^{(d,k)}$ over $\mathcal{H}_d^{\otimes k}$ as
	\begin{equation}
		P_{\mathrm{sym}}^{(d,k)} := \frac{1}{k!} \sum_{\pi \in S_k} V_\pi.
	\end{equation}
	
	\begin{fact}[{\cite[Theorem 22]{mele2024introduction}}]
		\label{fact:haar}
		For $n,k\in \mathbb N^+$ and $d=2^n$,
		\begin{equation}
			\EE{U\sim \mathrm{Haar}(n)}[U^{\otimes k}\ket{0}\bra{0}U^{\dagger\otimes k}]=\frac{P_{\mathrm{sym}}^{(d,k)}}{\tr(P_{\mathrm{sym}}^{(d,k)})}=\frac{\frac{1}{k!}\sum_{\pi\in S_k}V_d(\pi)}{\binom{k+d-1}{k}}.
		\end{equation}
	\end{fact}
	Consider a pure $n$-qubit state $\ket{\psi}$ on the composite system $AB$, with subsystem sizes $n_A$ and $n_B$. Performing computational basis measurements on subsystem $B$ yields the projected ensemble $\mathcal{E}_{\psi} \coloneqq \{(p_{\psi}(\bmz), \psi_{\bmz})\}$, where, writing $\psi\coloneqq\ketbra{\psi}{\psi}$:	
	\begin{equation}
		\begin{split}
			p_{\psi}(\bm z) &= \tr\big[(I_A \otimes \ketbra{\bm z}_B)\psi_{AB}\big], \\
			\ket{\psi_{\bm z}} &= \frac{(I_A \otimes \bra{\bm z}_B)\ket{\psi}_{AB}}{\sqrt{p_{\psi}(\bm z)}}.
		\end{split}
	\end{equation}
	We have:
	
	\begin{lemma}
		\label{lem:second_moment}
		For $p_{\psi}(\bm{z}) = \tr\left[(I_A \otimes \ketbra{\bm{z}}_B) \psi_{AB}\right]$, we have
		\begin{equation}
			\EE{\psi\sim\Haar(n)}\sum_{\bmz\in[d_B]}p_\psi(\bmz)^2=\frac{d_A+1}{d_Ad_B+1}.
		\end{equation}
	\end{lemma}
	\begin{proof}
		\begin{align}
			\EE{\psi\sim\Haar(n)}\sum_{\bmz\in[d_B]}p_\psi(\bmz)^2
			&=\EE{\psi\sim\Haar(n)}\sum_{\bmz\in[d_B]}\tr\big[(I_A\otimes\ket \bmz\bra\bmz_B)^{\otimes 2}\psi^{\otimes 2}\big]\\
			&=\frac{1}{(d_Ad_B+1)d_Ad_B}\sum_{\bmz\in[d_B]}\tr\big[(I_A\otimes\ket \bmz\bra\bmz_B)^{\otimes 2}(I+\mathbb{S})\big]\\
			&=\frac{1}{(d_Ad_B+1)d_A}\tr\big[(I_A\otimes\ket 0\bra0_B)^{\otimes 2}(I+\mathbb{S})\big] \label{eq:identical_summand}\\
			&=\frac{1}{(d_Ad_B+1)d_A}(d_A^2+d_A)\\
			&=\frac{d_A+1}{d_Ad_B+1}, 
		\end{align}
		where in Eq.~\eqref{eq:identical_summand} we use that the summand is identical for every $\bmz \in [d_B]$, so the sum over $z$ contributes an overall factor of $d_B$.
	\end{proof}

	\suppnote{Hardness of cloning ensembles} \label{app:proof_average}
	
	We begin by introducing the necessary notation. Let $d_A=2^{n_A}$ and $d_B=2^{n_B}$ be the dimensions of subsystems $A$ and $B$. When measuring the environment $B$ across $t$ copies of the input state $\ket{\psi}$, we denote the sequence of classical outcomes by $\bmz^{(t)} = (\bmz^{(t)}_1,\cdots,\bmz^{(t)}_t) \in [d_B]^t$. The probability of this sequence is given by $p_\psi(\bmz^{(t)})=\prod_{i=1}^t p_\psi(\bmz^{(t)}_i)$. The corresponding post-measurement state on the $t$ copies of subsystem $A$ is the product state $\psi_{\bmz^{(t)}}=\bigotimes_{i=1}^t \psi_{\bmz^{(t)}_i}$.
	
	We first present a detailed theorem that gives the explicit tradeoff between sample complexity and measurements for the ensemble cloning task.
	\begin{theorem}[No-cloning of quantum ensembles with tradeoff between sample-complexity and measurements]\label{thm:average_case_state_detailed}
		For every quantum channel $\Lambda$,
		\begin{equation}
			\mathbb{E}_{\psi \sim \Haar(n)}\ \TV{\Lambda(\psi^{\otimes t})}{\sigma_\psi^{(2)}} \ge  1-\sqrt{\frac{t(t-1)}{2}\frac{d_A+1}{d_Ad_B+1}} -\frac{2+4d_A}{(d_A+2)(d_A+1)}.
		\end{equation}
		Here, $D(\omega, \tau) \coloneqq \frac{1}{2} \Vert \omega - \tau \Vert_1$ denotes the trace distance.
	\end{theorem}
	If we require the average error to be small, e.g., $\mathbb{E}_{\psi \sim \Haar(n)} D(\Lambda(\psi^{\otimes t}), \sigma_\psi^{(2)}) \le 0.1$, then it necessitates that $t = \Omega(2^{n_B/2})$. This establishes the first part of Theorem~\ref{thm:summary}. 
	
	We now proceed to prove Theorem~\ref{thm:average_case_state_detailed}, which consists of four main steps. 
	\begin{enumerate}
		\item We first restrict to collision-free outcomes on $B$ across the $t$ copies of the purification $\psi$, and show that the contribution of collision events is negligible on average for Haar-random $\psi$. This is established in Lemma~\ref{lem:collision}.
		\item We then simplify the remaining state by removing irrelevant coherence. First, we eliminate the coherence between outcome sequences that cannot be transformed into one another by permutations. This is established in Lemma~\ref{lem:dephasing}.
		\item Next, Lemma~\ref{lem:reduction} removes the remaining coherence between outcome sequences that are related by permutations, and reduces the problem to a select-and-clone task for a product of independent Haar-random states on system $A$. Lemma~\ref{lem:symmetric} then further reduces this task to an average-case cloning problem for a single Haar-random state.
		
		\item Finally, we prove in Lemma~\ref{lem:no_cloning} an average-case no-cloning bound for a single Haar-random state. Combining this bound with the above reductions completes the proof of the theorem.
	\end{enumerate}
	We also discuss several possible generalizations of Theorem~\ref{thm:average_case_state_detailed} at the end of this section. 
	
	\subsection{Restricting to collision-free outcomes}
	
	The first step is to show that we can restrict our analysis to measurement outcomes that are collision-free.
	It suffices to consider the regime $t\le \sqrt{d_B}$, since for larger $t$ the lower bound in  Theorem~\ref{thm:average_case_state_detailed} is trivial. When $t$ is small, for a typical Haar-random state, the probability of a collision (i.e., a repeated outcome) on subsystem $B$ across $t$ copies is exponentially small in the environment size $n_B$. We therefore expect the original $t$-copy state $\psi^{\otimes t}$ to be very close to its projection onto the collision-free subspace.
	
	We formalize this by defining the set of collision-free outcome sequences as
	\begin{equation}
		R \coloneqq \{\bm z^{(t)}\in[d_B]^t : \bmz^{(t)}_i \neq \bmz^{(t)}_j \ \ \forall\,i\neq j\}.
	\end{equation}
	The projector onto this collision-free subspace of the environment is $\Pi_R\coloneqq I_A^{\otimes t} \otimes \sum_{\bmz^{(t)}\in R}\ket{\bmz^{(t)}}\bra{\bmz^{(t)}}_B$. We define the collision-free part of the state as the normalized projection:
	\begin{equation}
		\ket{\Phi^{(t)}_\psi} \coloneqq \frac{\Pi_R \ket{\psi}^{\otimes t}}{\|\Pi_R \ket{\psi}^{\otimes t}\|}.
	\end{equation}
	This state can be expressed as a superposition over only the collision-free outcomes:
	\begin{equation}
		\ket{\Phi^{(t)}_\psi}=\sum_{\bmz^{(t)}\in R}\sqrt{q_{\psi}(\bmz^{(t)})}\ket{\psi_{\bmz^{(t)}}}_A\ket{\bmz^{(t)}}_B.
	\end{equation}
	where 
	\begin{equation}\label{eq:prob_dist_no_collision}
		q_{\psi}(\bmz^{(t)}) \coloneqq \frac{p_{\psi}(\bmz^{(t)})}{\sum_{\bmz^{(t)} \in R} p_{\psi}(\bmz^{(t)})}
	\end{equation}
	is the conditional probability of obtaining the sequence $\bmz^{(t)}$ given that it is collision-free.
	
	The following lemma rigorously bounds the average distance between the original state and its collision-free part, confirming our intuition that they are exponentially close.
	\begin{lemma}
		\label{lem:collision}
		For $1 \le n_A \le n$, 
		\begin{equation}
			\EE{\psi\sim\Haar(n)}\left[\TV{\psi^{\ot t}}{\Phi^{(t)}_\psi}\right]\leq \sqrt{\frac{t(t-1)}{2}\frac{d_A+1}{d_Ad_B+1}}.
		\end{equation}
	\end{lemma}
	\begin{proof}
		Since both $\psi^{\ot t}$ and $\Phi^{(t)}_\psi$ are pure states, the trace distance can be written as
		\begin{equation}
			\begin{split}
				\TV{\psi^{\ot t}}{\Phi^{(t)}_\psi}&=\sqrt{1-|\langle\psi^{\ot t}|\Phi^{(t)}_\psi\rangle|^2} \\&=\sqrt{1-\sum_{\bmz^{(t)}\in R}p_{\psi}(\bmz^{(t)})}
				\\&=\sqrt{\sum_{\bmz^{(t)}\not\in R}p_{\psi}(\bmz^{(t)})}.
			\end{split}
		\end{equation}
		Now, we bound the sum over the collision outcomes. We have
		\begin{equation}\label{eq:collision_union}
			\begin{split}
				\sum_{\bmz^{(t)}\not\in R}p_{\psi}(\bmz^{(t)}) &\leq\sum_{1\leq i<j\leq t}\Pr_{\bmz^{(t)}\sim p_{\psi}}[\bmz^{(t)}_i=\bmz^{(t)}_j]\\
				&=\frac{t(t-1)}{2}\sum_{\bmz\in[d_B]}p_{\psi}(\bmz)^2,
		\end{split}\end{equation}
		where the inequality comes from the union bound. 
		Then, 
		\begin{equation}\begin{split}
				\EE{\psi\sim\Haar(n)}\TV{\psi^{\ot t}}{\Phi^{(t)}_\psi}&= \EE{\psi\sim\Haar(n)}\sqrt{\sum_{\bmz^{(t)}\not\in R}p_{\psi}(\bmz^{(t)})}\\
				&\leq \sqrt{\EE{\psi\sim\Haar(n)}\sum_{\bmz^{(t)}\not\in R}p_{\psi}(\bmz^{(t)})}\\
				&=\sqrt{\frac{t(t-1)}{2}\frac{d_A+1}{d_Ad_B+1}}.
		\end{split}\end{equation}
		where the inequality uses the concavity of the square-root function, and the last equality comes from Eq.~\eqref{eq:collision_union} and Lemma \ref{lem:second_moment}.     
	\end{proof}
	
	\subsection{Enforcing measurement and decoherence}
	While Lemma~\ref{lem:collision} reduces the problem to the cleaner setting in which the labels on $\mathcal{H}_B^{\otimes t}$ are collision-free, coherence between different labels may still remain.
	Now, we introduce a second simplification by showing that the coherence of sequences that cannot be transformed into each other via a permutation is useless. 
	This decoherence is achieved by measuring the environment register $B$, with measurement operators given by
	\begin{equation}
		\left\{\frac{1}{\sqrt{t!}}\sum_{\pi\in S_t}I_A^{\otimes t}\otimes \ketbra{\pi(\bmz^{(t)})}{\pi(\bmz^{(t)})}_B\right\}_{\bmz^{(t)}\in R}.
	\end{equation}
	After measurement, the unnormalized post-measurement state is given by
	\begin{equation}
		X_{\bmz^{(t)}}=\frac{q_\psi(\bmz^{(t)})}{t!}\sum_{\sigma,\pi\in S_t}\ket{\psi_{\sigma(\bmz^{(t)})}}\bra{\psi_{\pi(\bmz^{(t)})}}_A\ket{\sigma(\bmz^{(t)})}\bra{\pi(\bmz^{(t)})}_B.
	\end{equation}
	The dephased state is then defined as $\rho^{(t)}_\psi=\sum_{\bmz^{(t)}\in R}X_{\bmz^{(t)}}$. 
	
	The key intuition behind this simplification is that Haar randomness induces independent random phases for each label $\bmz^{(t)}$. Therefore, whenever two distinct labels $\bmz^{(t)}_1$ and $\bmz^{(t)}_2$ cannot be mapped to each other by a permutation, the off-diagonal term $\ket{\bmz^{(t)}_1}\bra{\bmz^{(t)}_2}$ acquires a random phase. Averaging over Haar-random states then eliminates such terms, thereby enforcing decoherence between these labels.
	The following lemma rigorously shows that any channel acting on $\Phi^{(t)}_\psi$ exhibits average-case performance no better than on $\rho^{(t)}_\psi$ for Haar-random $\psi$.
	\begin{lemma}
		\label{lem:dephasing}
		For every channel $\Lambda$, 
		\begin{equation}
			\EE{\substack{\psi\sim \mathrm{Haar}(n)}}  \TV{\Lambda(\rho^{(t)}_\psi)}{\sigma_\psi^{(2)}} \le \EE{\substack{\psi\sim \mathrm{Haar}(n)}}  \TV{\Lambda(\Phi^{(t)}_\psi)}{\sigma_\psi^{(2)}}.
		\end{equation}
	\end{lemma}
	\begin{proof}
		Consider a quantum state $\ket{\psi}$ and a phase vector $\vec\theta = (\theta(\bm z))_{\bm z \in [d_B]}$ with $\theta(\bm z) \in [0, 2\pi)$. By adding relative phases to each projected state, we define
		\begin{equation}
			\ket{\psi_{\vec\theta}} \coloneqq \sum_{\bmz}\sqrt{p_\psi(\bmz)}e^{\mathrm{i}\theta(\bmz)}\ket{\psi_{\bmz}}\ket{\bmz}.
		\end{equation}
		Note that for every $\vec\theta$ and $\bm z$, we have $\ketbra{\psi_{\vec\theta,\bm z}} = \ketbra{\psi_{\bm z}}$. Hence, $\sigma_\psi^{(2)} = \sigma_{\psi_{\vec\theta}}^{(2)}$. Therefore, we can replace the input state with a state that averages over phases:
		\begin{equation}\begin{split}
				\EE{\substack{\psi\sim \mathrm{Haar}(n)}} \TV{\Lambda(\Phi^{(t)}_\psi)}{\sigma_\psi^{(2)}} &= \EE{\substack{\psi\sim \mathrm{Haar}(n)}} \EE{\vec\theta} \TV{\Lambda(\Phi^{(t)}_{\psi_{\vec\theta}})}{\sigma_\psi^{(2)}}\\
				&\geq \EE{\substack{\psi\sim \mathrm{Haar}(n)}} \TV{\Lambda\left(\EE{\vec\theta}\Phi^{(t)}_{\psi_{\vec\theta}}\right)}{\sigma_\psi^{(2)}}.
		\end{split}\end{equation}
		where the first equality comes from the invariance of the Haar measure under adding relative phases, and the second comes from Jensen's inequality. 
		
		Now we show that $\mathbb{E}_{\vec\theta}\Phi^{(t)}_{\psi_{\vec\theta}}$ is just $\rho_\psi^{(t)}$. Expanding the state, we have
		\begin{equation}\begin{split}
				\EE{\vec\theta} \Phi^{(t)}_{\psi_{\vec\theta}}
				&= \sum_{\bmz^{(t)},\tilde\bmz^{(t)}\in R}
				\sqrt{q_{\psi}(\bmz^{(t)})q_{\psi}(\tilde\bmz^{(t)})} \Bigl(\EE{\vec\theta} 
				e^{\mathrm{i}\sum_{i}[\theta(\bmz^{(t)}_i)-\theta(\tilde\bmz^{(t)}_i)]}\Bigr) \ketbra{\psi_{\bmz^{(t)}}}{\psi_{\tilde\bmz^{(t)}}}_A\ketbra{\bmz^{(t)}}{\tilde\bmz^{(t)}}_B.
		\end{split}\end{equation}
		If there exists an index $\bm z \in [d_B]$ such that $\bm z \in \bm z^{(t)}$ but $\bm z \notin \tilde{\bm z}^{(t)}$, then the coefficient $\bE_{\vec\theta} e^{\text{i}\sum_i[\theta(\bmz^{(t)}_i)-\theta(\tilde \bmz^{(t)}_i)]}$ will be $0$. 
		The only surviving terms are those where $\bm z^{(t)} = \pi(\tilde{\bm z}^{(t)})$ for some permutation $\pi \in S_t$. In this case, we have $\bE_{\vec\theta} e^{\text{i}[\sum_i\theta(\bmz^{(t)}_i)-\theta(\tilde \bmz^{(t)}_i)]}=1$. Therefore,
		\begin{equation}\begin{split}
				\EE{\vec\theta} \Phi^{(t)}_{\psi_{\vec\theta}}
				&= \sum_{\bmz^{(t)}\in R}\sum_{\pi\in S_t}
				q_{\psi}(\bmz^{(t)}) \ketbra{\psi_{\bmz^{(t)}}}{\psi_{\pi(\bmz^{(t)})}}_A\ketbra{\bmz^{(t)}}{\pi(\bmz^{(t)})}_B\\
				&= \sum_{\bmz^{(t)}\in R}\sum_{\sigma,\pi\in S_t}
				\frac{q_{\psi}(\bmz^{(t)})}{t!} \ketbra{\psi_{\sigma(\bmz^{(t)})}}{\psi_{\pi(\bmz^{(t)})}}_A\ketbra{\sigma(\bmz^{(t)})}{\pi(\bmz^{(t)})}_B\\
				&=\rho^{(t)}_\psi.
		\end{split}\end{equation}
		This completes the proof.
	\end{proof}
	
	\subsection{Reducing to cloning a random state}
	We now show that it is impossible to prepare $\sigma_{\psi}^{(2)}$ given $\rho^{(t)}_\psi$. 
	The key intuition is that, given the dephased input $\rho^{(t)}_{\psi}$, we can only sample a no-collision sequence $\bm z^{(t)}$ with the corresponding state $\ket{\psi_{\bmz^{(t)}}}$, meaning we only have access to a single copy of each projected state. However, the goal requires producing two identical copies of these states. This turns the problem into an average-case cloning task.	
	
	This reduction is achieved via the following Lemmas \ref{lem:reduction} and \ref{lem:symmetric}. 
	We begin with Lemma~\ref{lem:reduction}. Although $\rho_{\psi}^{(t)}$ has already been dephased between labels $\bmz^{(t)}$ that are not related by permutations, it still retains coherence between different labels that can be transformed into each other via permutations. Lemma~\ref{lem:reduction} shows that even this remaining coherence can be removed. The key idea is to construct another channel $\Lambda'$ that takes $d_B$ independent Haar-random product states as input and uses them to simulate the action of $\Lambda$ on $\rho_{\psi}^{(t)}$. This reduces the problem to a select-and-clone task: one of these independent input states is selected, its label is recorded in an ancilla register, and the selected state is then cloned.
	
	\begin{lemma}
		\label{lem:reduction}
		For every $1 \le t \leq d_B$ and channel $\Lambda$, if
		\begin{equation}
			\EE{\substack{\psi\sim \mathrm{Haar}(n)}}  \TV{\Lambda(\rho^{(t)}_\psi)}{\sigma_\psi^{(2)}}
			\le \delta,
		\end{equation}
		then there exists a channel $\Lambda'$ such that
		\begin{equation}
			\label{eq:cond1}
			\EE{\phi_1,\cdots,\phi_{d_B}\sim \mathrm{Haar}(n_A)}
			\Biggl[\sum_{i=1}^{d_B}
			\bra{\phi_i}^{\otimes 2}\bra{i}\;
			\Lambda'\bigl(\bigotimes_{j=1}^{d_B}\phi_j\bigr)\;
			\ket{\phi_i}^{\otimes 2}\ket{i}
			\Biggr]
			\ge 1-\delta.
		\end{equation}
	\end{lemma}
	\begin{proof}
		For any state $\tau$ on system $AB$, we use the variational expression for the trace distance, which gives:
		\begin{equation}\begin{split}
				\TV{\tau}{\sigma_\psi^{(2)}}
				&=\max_{0\le P\le I}\tr\bigl(P(\sigma_\psi^{(2)}-\tau)\bigr)\\
				&\ge \sum_{\bmz\in[d_B]} \bra{\psi_\bmz}^{\otimes 2}\bra{\bmz}\sigma_\psi^{(2)}\ket{\psi_\bmz}^{\otimes 2}\ket{\bmz}
				-\sum_{\bmz\in[d_B]} \bra{\psi_\bmz}^{\otimes 2}\bra{\bmz}\tau\ket{\psi_\bmz}^{\otimes 2}\ket{\bmz}\\
				&= 1-\sum_{\bmz\in[d_B]} \bra{\psi_\bmz}^{\otimes 2}\bra{\bmz}\tau\ket{\psi_\bmz}^{\otimes 2}\ket{\bmz}.
		\end{split}\end{equation}
		where the inequality is obtained by choosing $P =  \sum_\bmz\ket{\psi_\bmz}\bra{\psi_\bmz}^{\otimes 2}\ket{\bmz}\bra{\bmz}$. 
		Substituting $\tau=\Lambda(\rho^{(t)}_{\psi})$, we obtain
		\begin{equation}
			\label{eq:cond2}
			\EE{\psi\sim\mathrm{Haar}(n)}
			\sum_{\bmz\in[d_B]} \bra{\psi_\bmz}^{\otimes 2}\bra{\bmz}\Lambda(\rho^{(t)}_\psi)\ket{\psi_\bmz}^{\otimes 2}\ket{\bmz}
			\ge 1-\delta.
		\end{equation}
		
		Now we construct the channel $\Lambda'$ satisfying Eq.~(\ref{eq:cond1}). The idea is that $\Lambda'$ first maps its input $\phi_1,\cdots,\phi_{d_B}$ to a state $\tilde \rho$ that is distributed in the same way as the input state $\rho_\psi^{(t)}$ of $\Lambda$. It then applies $\Lambda$ to $\tilde \rho$ and outputs the result. By Eq. (\ref{eq:cond2}), the channel $\Lambda'$ satisfies the desired condition.
		
		More explicitly, the channel $\Lambda'$ is constructed as follows:
		\begin{enumerate}
			\item $\Lambda'$ takes the product state $\bigotimes_{i=1}^{d_B}\phi_i$ as input.
			
			\item Let $\nu$ denote the distribution over probability distributions $p_\psi$ induced by a Haar-random state $\psi$. In other words, $p_0 \sim \nu$ means that one first draws $\psi$ from the Haar measure and then sets $p_0 = p_\psi$. The channel $\Lambda'$ samples a probability distribution $p_0$ from $\nu$, and computes the corresponding $q_0$ as given in Eq.~\eqref{eq:prob_dist_no_collision}. This step simulates the distribution function $q_\psi$ associated with $\rho_\psi^{(t)}$.
			
			\item Sample a sequence $\bm z^{(t)} \sim q_0$ from the set $R$. This step simulates the set of classical registers of $\rho_\psi^{(t)}$.
			
			\item Prepare the uniform superposition $\frac{1}{\sqrt{t!}} \sum_{\sigma \in S_t} \ket{\sigma}$ over the permutation group $S_t$, and use $\sigma$ to permute the states $\{\phi_i\}_i$ and labels $\bmz^{(t)}$, resulting in the state
			\begin{equation}
				\frac{1}{\sqrt{t!}}\sum_{\sigma\in S_t}\ket{\sigma}\ket{\phi_{\sigma(\bmz^{(t)})}}\ket{\sigma(\bmz^{(t)})}.
			\end{equation}
			This step constructs a simulated version of $\rho_\psi^{(t)}$, up to an entangled ancillary register $\ket{\sigma}$.
			
			\item Erase the $\ket{\sigma}$ register using $\ket{\sigma(\bmz^{(t)})}$ and $\bmz^{(t)}$. In particular, since the entries of $\bmz^{(t)}$ are pairwise distinct, one can construct a circuit $f_{\bmz^{(t)}}$ that computes $\sigma$ from $\sigma(\bmz^{(t)})$ given $\bmz^{(t)}$. The circuit $f_{\bmz^{(t)}}$ acts as
			\begin{equation}
				\ket{\sigma(\bmz^{(t)})}\ket{y}\mapsto \ket{\sigma(\bmz^{(t)})}\ket{y\oplus \sigma}.
			\end{equation}
			We then apply $f_{\bmz^{(t)}}$ to the third and first registers, which resets the first register to $\ket{0}$. Denote the resulting state by
			\begin{equation}
				\ket{\tilde\rho}:=\frac{1}{\sqrt{t!}}\sum_{\sigma\in S_t}\ket{\phi_{\sigma(\bmz^{(t)})}}\ket{\sigma(\bmz^{(t)})}.
			\end{equation}
			The state $\tilde\rho$ is now distributed in the same way as the input $\rho_\psi^{(t)}$ of $\Lambda$.
			
			\item Finally, apply $\Lambda$ to $\tilde\rho$ and output the resulting state:
			\begin{equation}
				\Lambda(\tilde\rho)=\Lambda\left(\frac{1}{t!}\sum_{\sigma,\pi\in S_t}\ketbra{\phi_{\sigma(\bmz^{(t)})}}{\phi_{\pi(\bmz^{(t)})}}_A\ketbra{\sigma(\bmz^{(t)})}{\pi(\bmz^{(t)})}_B\right).
			\end{equation}
		\end{enumerate}
		The result of $\Lambda'$ can be written as
		\begin{equation}\begin{split}
				\Lambda'(\bigotimes_{j=1}^{d_B}\phi_j)&\coloneqq \EE{\tilde\rho} \Lambda(\tilde\rho)\\
				&=\Lambda\left(\EE{p_0\sim\nu}\ \sum_{\bmz^{(t)}\in R}\frac{q_0(\bmz^{(t)})}{t!}\sum_{\sigma,\pi\in S_t}\ketbra{\phi_{\sigma(\bmz^{(t)})}}{\phi_{\pi(\bmz^{(t)})}}_A\ketbra{\sigma(\bmz^{(t)})}{\pi(\bmz^{(t)})}_B\right).
		\end{split}\end{equation}
		We now verify that $\Lambda'$ indeed satisfies the claimed condition. 
		\begin{align}
			&\EE{\phi_1,\cdots,\phi_{d_B}\sim \mathrm{Haar}(n_A)}[
			\sum_{l\in[d_B]}\bra{\phi_l}^{\ot 2}\bra l\Lambda'(\bigotimes_{j=1}^{d_B}\ket{\phi_j}\bra{\phi_j})\ket{\phi_l}^{\ot 2}\ket{l}
			]\\
			&=\EE{\phi_1,\cdots,\phi_{d_B}\sim \mathrm{Haar}(n_A)}
			\sum_{l\in[d_B]}\bra{\phi_l}^{\ot 2}\bra{l}
			\Lambda\left(\EE{p_0\sim\nu}\ \sum_{\bmz^{(t)}\in R}\frac{q_0(\bmz^{(t)})}{t!}\sum_{\sigma,\pi\in S_t}\ketbra{\phi_{\sigma(\bmz^{(t)})}}{\phi_{\pi(\bmz^{(t)})}}_A\ketbra{\sigma(\bmz^{(t)})}{\pi(\bmz^{(t)})}_B\right)
			\ket{\phi_l}^{\ot 2}\ket{l}\\
			&=\EE{\phi_1,\cdots,\phi_{d_B}\sim \mathrm{Haar}(n_A)}
			\EE{p_0\sim\nu}\ \sum_{l\in[d_B]}\bra{\phi_l}^{\ot 2}\bra{l} 
			\Lambda\left(\sum_{\bmz^{(t)}\in R}\frac{q_0(\bmz^{(t)})}{t!}\sum_{\sigma,\pi\in S_t}\ketbra{\phi_{\sigma(\bmz^{(t)})}}{\phi_{\pi(\bmz^{(t)})}}_A\ketbra{\sigma(\bmz^{(t)})}{\pi(\bmz^{(t)})}_B\right)
			\ket{\phi_l}^{\ot 2}\ket{l}\\
			&=\EE{\psi\sim \mathrm{Haar}(n)}
			\sum_{l\in[d_B]}\bra{\psi_l}^{\ot 2}\bra{l}
			\Lambda\left(\sum_{\bmz^{(t)}\in R}\frac{q_\psi(\bmz^{(t)})}{t!}\sum_{\sigma,\pi\in S_t}\ketbra{\psi_{\sigma(\bmz^{(t)})}}{\psi_{\pi(\bmz^{(t)})}}_A\ketbra{\sigma(\bmz^{(t)})}{\pi(\bmz^{(t)})}_B\right)
			\ket{\psi_l}^{\ot 2}\ket{l} \label{eq:nu_Haar_random}\\
			&=\EE{\psi\sim \mathrm{Haar}(n)}
			\sum_{l\in[d_B]}\bra{\psi_l}^{\ot 2}\bra{l}
			\sum_{\bmz^{(t)}\in R}\Lambda\left(X_{\bmz^{(t)}}\right)
			\ket{\psi_l}^{\ot 2}\ket{l}\\
			&=\EE{\psi\sim \mathrm{Haar}(n)}
			\sum_{l\in[d_B]}\bra{\psi_l}^{\ot 2}\bra{l}
			\Lambda\left(\rho_\psi^{(t)}\right)
			\ket{\psi_l}^{\ot 2}\ket{l}\\
			&\geq \, 1-\delta.
		\end{align}
		Here, Eq.~\eqref{eq:nu_Haar_random} follows by setting $\ket{\psi}\coloneqq \sum_{\bmz\in[d_B]}\sqrt{p_0(\bmz)}\ket{\phi_\bmz}_A\ket{\bmz}_B$ and using the definition of $\nu$. This completes the proof.
	\end{proof}
	
	Lemma~\ref{lem:symmetric} shows that any channel achieving nontrivial performance on this select-and-clone task for $d_B$ independent Haar-random states can be converted into another channel with the same average performance for cloning a single Haar-random state.
	\begin{lemma}
		\label{lem:symmetric}
		For every $\delta\in(0,1)$, if there exists a channel $\Lambda$ such that
		\begin{equation}\label{eq:symmetric_expectation}
			\EE{\phi_1,\cdots,\phi_{d_B}\sim \mathrm{Haar}(n_A)}[\sum_{i=1}^{d_B}\bra{\phi_i}^{\ot 2}\bra i\Lambda(\bigotimes_{j=1}^{d_B}\phi_j)\ket{\phi_i}^{\ot 2}\ket{i}]\geq \delta,
		\end{equation}
		then there exists a channel $\Lambda'$ such that
		\begin{equation}\label{eq:symmetric_requirement}
			\EE{\phi\sim \mathrm{Haar}(n_A)}[\bra{\phi}^{\ot 2} \Lambda'(\phi)\ket{\phi}^{\ot 2}]\geq \delta.
		\end{equation}
	\end{lemma}
	
	\begin{proof}
		Let $S_{d_B}$ be the permutation group over $[d_B]$. For $\pi\in S_{d_B}$, define the permutation channels
		\begin{equation}\begin{split}
				\mathcal V_\pi:\bigotimes_{i=1}^{d_B}\phi_i
				&\mapsto \bigotimes_{i=1}^{d_B}\phi_{\pi^{-1}(i)},\\
				\mathcal W_\pi:\ketbra{i}{i} &\mapsto \ketbra{\pi(i)}{\pi(i)}.
		\end{split}\end{equation}
		
		Given a channel $\Lambda$, define
		\begin{equation}
			\Gamma_\pi(\Lambda) := \mathcal W_\pi \circ \Lambda \circ \mathcal V_\pi .
		\end{equation}
		By a change of variables $\phi_i' \coloneqq \phi_{\pi^{-1}(i)}$, we see that the expectation in Eq.~\eqref{eq:symmetric_expectation} is invariant when replacing $\Lambda$ with $\Gamma_\pi(\Lambda)$.
		
		Let $U(\mathcal{H})$ be the set of unitaries acting on $\mathcal{H}$. For a sequence of unitary operators ${\rm U}=(U_1,\ldots,U_{d_B})\in U({\cal H}_A)^{d_B}$, denote by $\mathcal U_i:\sigma\mapsto U_i\sigma U_i^\dagger$ the unitary channels induced by $U_i$, and the adjoint unitary channel $\mathcal U_i^\dagger:\sigma\mapsto U_i^\dagger\sigma U_i$ induced by $U_i^\dagger$.
		Denote $\mathcal E_i:\sigma\mapsto \ketbra{i}{i}\sigma\ketbra{i}{i}$.
		Define the channel
		\begin{equation}
			\Theta_{\rm U}(\Lambda)
			:= \biggl(\sum_{i=1}^{d_B} \mathcal U_i^\dagger \otimes \mathcal U_i^\dagger \otimes \mathcal E_i\biggr)
			\circ \Lambda \circ \bigotimes_{i=1}^{d_B} \mathcal U_i.
		\end{equation}
		Similarly, by defining $\ket{\phi_i'} \coloneqq U_i\ket{\phi_i}$,  we see that the expectation in Eq.~\eqref{eq:symmetric_expectation} remains invariant when replacing $\Lambda$ with $\Theta_{\rm U}(\Lambda)$.
		
		Averaging over both symmetries, define
		\begin{equation}
			\bar \Lambda:=\frac{1}{d_B!}\sum_{\pi\in S_{d_B}}\EE{\mathrm U\sim \Haar(n_A)^{d_B}} \Theta_{\rm U}(\Gamma_\pi(\Lambda)).
		\end{equation}
		Then $\bar\Lambda$ satisfies $\Gamma_\pi(\bar\Lambda)=\bar\Lambda$ and $\Theta_{\rm U}(\bar\Lambda)=\bar\Lambda$ for all $\pi,{\rm U}$, while preserving the expectation in Eq.~\eqref{eq:symmetric_expectation}. Concretely,
		\begin{equation}\begin{split}
				\label{eq:twirling}
				&\EE{\phi_1,\cdots,\phi_{d_B}\sim \mathrm{Haar}(n_A)}[\sum_{i=1}^{d_B}\bra{\phi_i}^{\ot 2}\bra i\bar\Lambda(\bigotimes_{j=1}^{d_B}\phi_j)\ket{\phi_i}^{\ot 2}\ket{i}]
				=\EE{\phi_1,\cdots,\phi_{d_B}\sim \mathrm{Haar}(n_A)}[\sum_{i=1}^{d_B}\bra{\phi_i}^{\ot 2}\bra i\Lambda(\bigotimes_{j=1}^{d_B}\phi_j)\ket{\phi_i}^{\ot 2}\ket{i}].
		\end{split}\end{equation}
		
		For $i\in[d_B]$, define a map
		\begin{equation}
			\label{eq:twirled}
			\Psi_i(\sigma) :=
			\tr_3\Bigl[(I_A\otimes I_A \otimes \ketbra{i})\;
			\bar\Lambda\bigl((I_A/d_A)^{\otimes (i-1)}\otimes \sigma \otimes (I_A/d_A)^{\otimes (d_B-i)}\bigr)\Bigr].
		\end{equation}
		Each $\Psi_i$ is completely positive, and the permutation invariance of $\bar\Lambda$ implies $\Psi_i=\Psi_j$ for all $i,j\in[d_B]$.
		
		Let $\mu^{(i)}$ be a distribution over $U(\mathcal H_A)^{d_B}$, where each $U_j$ independently forms $\Haar(n_A)$ for $j\neq i$, and $U_i$ is always the identity operator $I_A$. Then, we have
		\begin{align}
			&\EE{\phi_1,\cdots,\phi_{d_B}\sim \mathrm{Haar}(n_A)}[\sum_{i=1}^{d_B}\bra{\phi_i}^{\ot 2}\bra i\bar \Lambda(\bigotimes_{j=1}^{d_B}\phi_j)\ket{\phi_i}^{\ot 2}\ket{i}]\\
			&=\EE{\phi_1,\cdots,\phi_{d_B}\sim \mathrm{Haar}(n_A)}[\sum_{i=1}^{d_B}\bra{\phi_i}^{\ot 2}\bra i \EE{\mathrm{U}\sim \mu^{(i)}}[\Theta_{\rm U}(\bar\Lambda)(\bigotimes_{j=1}^{d_B}\phi_j)]\ket{\phi_i}^{\ot 2}\ket{i}] \label{eq:unitary_invariable_lambda}\\
			&=\EE{\phi_1,\cdots,\phi_{d_B}\sim \mathrm{Haar}(n_A)}[\sum_{i=1}^{d_B}\bra{\phi_i}^{\ot 2}\bra i \EE{\mathrm{U}\sim \mu^{(i)}}[\biggl(\sum_{j=1}^{d_B} \mathcal U_j^\dagger \otimes \mathcal U_j^\dagger \otimes \mathcal E_j\biggr)
			\circ \bar\Lambda (\bigotimes_{j=1}^{d_B}U_j\phi_jU_j^\dagger)]\ket{\phi_i}^{\ot 2}\ket{i}]\\   
			&=\EE{\phi_1,\cdots,\phi_{d_B}\sim \mathrm{Haar}(n_A)}[\sum_{i=1}^{d_B}\bra{\phi_i}^{\ot 2}\bra i \EE{\mathrm{U}\sim \mu^{(i)}}[\bar\Lambda (\bigotimes_{j=1}^{d_B}U_j\phi_jU_j^\dagger)]\ket{\phi_i}^{\ot 2}\ket{i}]\label{eq:mu_i_twirling}\\   
			&=\EE{\phi_1,\cdots,\phi_{d_B}\sim \mathrm{Haar}(n_A)}[\sum_{i=1}^{d_B}\bra{\phi_i}^{\ot 2}\bra i \bar\Lambda((I_A/d_A)^{\otimes (i-1)}\otimes \phi_i\otimes (I_A/d_A)^{\otimes (d_B-i)})]\ket{\phi_i}^{\ot 2}\ket{i}]\label{eq:mu_i_twirling2}\\
			&=\EE{\phi_1,\cdots,\phi_{d_B}\sim \mathrm{Haar}(n_A)}[\sum_{i=1}^{d_B}\bra{\phi_i}^{\ot 2}\Psi_i(\phi_i)\ket{\phi_i}^{\ot 2}] \\
			&= d_B \EE{\phi_1\sim \mathrm{Haar}(n_A)}[\bra{\phi_1}^{\ot 2}\Psi_1(\phi_1)\ket{\phi_1}^{\ot 2}] \label{eq:equivalence_Psi}\\
			&\geq \delta.
		\end{align}
		Here, Eq.~\eqref{eq:unitary_invariable_lambda} comes from the unitary invariance of $\bar\Lambda$. Eq.~\eqref{eq:mu_i_twirling} comes from $U_i=I_A$ and $\bra i\mathcal E_j(\rho)\ket i\equiv 0$ for every $\rho$ and $j\neq i$. Eq.~\eqref{eq:mu_i_twirling2} comes from $\bE_{U_j\sim \Haar(n_A)}[U_j\phi_jU_j^{\dagger}]=I_A/d_A$ for every $j\neq i$. Eq~\eqref{eq:equivalence_Psi} comes from the equivalence of all $\Psi_i$. Therefore, 
		\begin{equation}
			\EE{\phi\sim \mathrm{Haar}(n_A)}[\bra{\phi}^{\ot 2}\Psi_1(\phi)\ket{\phi}^{\ot 2}]\geq \delta/d_B.
		\end{equation}
		
		We construct $\Lambda'=d_B\cdot\Psi_1$. Then, $\Lambda'$ satisfies Eq.~\eqref{eq:symmetric_requirement}. 
		The remaining task is to prove that $\Lambda'$ is indeed a channel. The complete positivity of $\Lambda'$ follows from the complete positivity of $\Psi_1$. To check trace preservation, note that unitary invariance of $\bar \Lambda$ implies
		\begin{equation}
			\Psi_1(\mathcal{U}(\sigma)) = (\mathcal{U} \otimes \mathcal{U}) (\Psi_1(\sigma)), \quad \text{for every } U\in U(\mathcal{H}_A).
		\end{equation}
		Hence for every state $\sigma\in {\cal D}({\cal H}_A)$,
		\begin{equation}\label{eq:trace_preseve}\begin{split}
				\tr\left[\Psi_1(\sigma)\right]
				&=\tr\left[\EE{U\sim \Haar(n_A)}\ (U^{\dagger}\otimes U^{\dagger})\ \Psi_1(U\sigma U^{\dagger})\ (U\otimes U) \right]\\
				&=\EE{U\sim \Haar(n_A)}\ \tr\left[(U^{\dagger}\otimes U^{\dagger})\ \Psi_1(U\sigma U^{\dagger})\ (U\otimes U)\right]\\
				&=\EE{U\sim \Haar(n_A)}\ \tr\left[\Psi_1(U\sigma U^{\dagger})\right]\\
				&=\tr\left[\Psi_1\left(\EE{U\sim \Haar(n_A)}\ U\sigma U^{\dagger}\right)\right]\\
				&=\tr\left[\Psi_1(I_A/d_A)\right].
		\end{split}\end{equation}
		By equivalence of all $\Psi_i$ and the definition of $\Psi_i$,
		\begin{equation}\begin{split}
				\tr\left[\Psi_1(I_A/d_A)\right]
				&=\frac{1}{d_B}\tr\left[\sum_{i=1}^{d_B}\Psi_i(I_A/d_A)\right]\\
				&=\frac{1}{d_B}\tr\left[\bar\Lambda\bigl((I_A/d_A)^{\otimes d_B}\bigr)\right]\\
				&=\frac{1}{d_B},
		\end{split}\end{equation}
		Therefore, 
		\begin{equation}
			\tr\left[\Lambda'(\sigma)\right] = d_B \tr\left[\Psi_1(I_A/d_A)\right] = 1, \quad \text{for every } \sigma\in {\cal D}({\cal H}_A).
		\end{equation}
		This completes the proof.
	\end{proof}
	
	\subsection{Hardness of cloning a random state}
	Finally, we establish that any quantum channel cannot clone a random state with high precision. This is stated in the following lemma:
	\begin{lemma}\label{lem:no_cloning}
		For every channel $\Lambda$,
		\begin{equation}
			\EE{\psi\sim \mathrm{Haar}(n_A)}[\bra{\psi}^{\ot 2} \Lambda(\psi)\ket{\psi}^{\ot 2}]\leq \frac{2+4d_A}{(d_A+2)(d_A+1)}.
		\end{equation}
	\end{lemma}

	\begin{proof}
		Let $J(\Lambda)$ be the Choi-Jamiołkowski representation of the channel $\Lambda$ defined as
		\begin{equation}\begin{split}
				J(\Lambda)\coloneqq \sum_{i,j\in [d_A]}(\ket{i}\bra{j})_A\otimes (\Lambda(\ket{i}\bra{j}))_C.
		\end{split}\end{equation}
		For every $\rho$ on system $A$, we have
		\begin{equation}\begin{split}
				\Lambda(\rho)&=\sum_{i,j\in[d_A]}\Lambda(\ket{i}\bra{i}\rho\ket{j}\bra{j})\\
				&=\sum_{i,j\in[d_A]}\bra{i}\rho\ket{j}\tr_A(J(\Lambda)(\ket{j}\bra{i}\otimes I_C))\\
				&=\tr_A[J(\Lambda)(\rho^T\otimes I_C)].
		\end{split}\end{equation}
		For every $\psi\in {\cal H}_A$, we then have
		\begin{equation}\begin{split}
				\bra{\psi}^{\otimes 2}\Lambda(\psi)\ket{\psi}^{\otimes 2}
				&=\bra{\psi}^{\otimes 2}\tr_A[J(\Lambda)(\psi^T\otimes I_C)]\ket{\psi}^{\otimes 2}\\
				&=\tr[J(\Lambda)(\psi^T\otimes \psi^{\otimes 2})].
		\end{split}\end{equation}
		
		Define $M \coloneqq \bE_{\psi\sim \Haar(n_A)} \psi^{T} \otimes \psi^{\otimes 2}$. Then the quantity we aim to upper bound is
		\begin{equation}\begin{split}
				\mathbb{E}_{\psi \sim \Haar(n_A)}\!\left[\bra{\psi}^{\otimes 2}\Lambda(\psi)\ket{\psi}^{\otimes 2}\right]
				&= \mathbb{E}_{\psi \sim \Haar(n_A)} \tr\!\left(J(\Lambda)\bigl(\psi^{T} \otimes \psi^{\otimes 2}\bigr)\right)\\
				&=\tr(J(\Lambda) M)\\
				&\leq \|J(\Lambda)\|_1\cdot \|M\|_\infty\\
				&=d_A\cdot \|M\|_\infty.
		\end{split}\end{equation}
		Denote $T_1$ as the partial transpose acting on the first subsystem only. By Fact \ref{fact:haar}, we obtain
		\begin{equation}\begin{split}
				M &=\mathbb{E}_{\psi \sim \Haar(n_A)} \left[\psi^{\otimes 3} \right]^{T_1}\\
				&=\left[\frac{P_{\text{sym}}^{(d_A,3)}}{\tr(P_{\text{sym}}^{(d_A,3)})}\right]^{T_1}\\
				&=\frac{1}{{\binom{d_A+2}{3}}}\left(P_{\text{sym}}^{(d_A,3)}\right)^{T_1}.
		\end{split}\end{equation}
		Thus, we have
		\begin{equation}\begin{split}
				\|M\|_\infty&=\frac{1}{{\binom{d_A+2}{3}}}\left\|\left(P_{\text{sym}}^{(d_A,3)}\right)^{T_1}\right\|_\infty\\
				&\leq\frac{1}{(d_A+2)(d_A+1)d_A}\sum_{\pi\in S_3}\|V_\pi^{T_1}\|_\infty.
		\end{split}\end{equation}
		where $V_\pi$ is the permutation unitary defined by
		\begin{equation}\begin{split}
				V_\pi(\ket{\psi_1}\ket{\psi_2}\ket{\psi_3})=\ket{\psi_{\pi^{-1}(1)}}\ket{\psi_{\pi^{-1}(2)}}\ket{\psi_{\pi^{-1}(3)}}
		\end{split}\end{equation}
		for every $\ket{\psi_1},\ket{\psi_2},\ket{\psi_3}\in {\cal H}_A$. It remains to bound $\|V_\pi\|_\infty$ for each $\pi\in S_3$.
		
		For the case when $\pi(1)=1$, $V_\pi$ can be written as $V_\pi=I_1\otimes W_{23}$ for some operator $W_{23}$. Hence,
		\begin{equation}\begin{split}
				\|V_\pi^{T_1}\|_\infty=\|I_1^T\otimes W_{23}\|_\infty=\|I_1\otimes W_{23}\|_\infty=\|V_\pi\|_\infty=1.
		\end{split}\end{equation}
		
		For $\pi=(1\ 2)$, we have
		\begin{equation}\begin{split}
				V_{(1\ 2)}^{T_1}&=\left(\sum_{i,j\in [d_A]}\ket{i}_1\ket{j}_2\bra{j}_1\bra{i}_2\right)^{T_1}\otimes I_3\\
				&=\left(\sum_{i,j\in [d_A]}\ket{j}_1\ket{j}_2\bra{i}_1\bra{i}_2\right)\otimes I_3\\
				&=(\ketbra{\Phi})_{12}\otimes I_3.
		\end{split}\end{equation}
		Here, $\ket{\Phi}=\sum_{i\in[d_A]} \ket i\ket i$ is the unnormalized maximally entangled state. Therefore, 
		\begin{equation}\begin{split}
				\|V_{(1\ 2)}^{T_1}\|_\infty=\|\ketbra{\Phi}\|_\infty=d_A.
		\end{split}\end{equation}
		Similarly, we also have $\|V_{(1\ 3)}^{T_1}\|_\infty=d_A$.
		
		For $\pi=(1\ 2\ 3)$, we have
		\begin{equation}\begin{split}
				V_{(1\ 2\ 3)}^{T_1}&=\left(\sum_{i,j,t\in[d_A]}\ket{i}_1\ket{j}_2\ket{t}_3\bra{t}_1\bra{i}_2\bra{j}_3\right)^{T_1}\\
				&=\sum_{i,j,t\in[d_A]}\ket{t}_1\ket{j}_2\ket{t}_3\bra{i}_1\bra{i}_2\bra{j}_3\\
				&=V_{(2\  3)}\sum_{i,j,t\in[d_A]}\ket{t}_1\ket{t}_2\ket{j}_3\bra{i}_1\bra{i}_2\bra{j}_3\\
				&=V_{(2\ 3)}(\ketbra{\Phi})_{12}\otimes I_3.
		\end{split}\end{equation}
		Hence, 
		\begin{equation}\begin{split}
				\norm{V_{(1\ 2\ 3)}^{T_1}}_\infty=\norm{\ketbra{\Phi}}_\infty=d_A.
		\end{split}\end{equation}
		Similarly, we also have $\|V_{(1\ 3\ 2)}^{T_1}\|_\infty=d_A$.
		
		Combining the above bounds, we have
		\begin{equation}\begin{split}
				\EE{\phi\sim \mathrm{Haar}(n_A)}[\bra{\phi}^{\ot 2} \Lambda(\phi)\ket{\phi}^{\ot 2}]
				&\leq d_A\cdot \frac{1}{(d_A+2)(d_A+1)d_A}\sum_{\pi\in S_3}\|V_\pi^{T_1}\|_\infty\\
				&=\frac{2+4d_A}{(d_A+2)(d_A+1)}.\\
		\end{split}\end{equation}
	\end{proof}
	Combining the above lemmas, we can now proceed to prove Theorem \ref{thm:average_case_state_detailed}.
	\begin{proof}[Proof of Theorem \ref{thm:average_case_state_detailed}]
		For any channel $\Lambda$ such that
		\begin{equation}
			\EE{\psi\sim\Haar(n)}\TV{\Lambda(\psi^{\ot t})}{\sigma^{(2)}_{\psi}}\leq \delta,
		\end{equation}
		we have
		\begin{equation}\begin{split}
				\EE{\psi\sim\Haar(n)}\TV{\Lambda(\Phi^{(t)}_\psi)}{\sigma^{(2)}_{\psi}}&\leq \EE{\psi\sim\Haar(n)}\TV{\Lambda(\Phi^{(t)}_\psi)}{\Lambda(\psi^{\otimes t})}+\TV{\Lambda(\psi^{\otimes t})}{\sigma^{(2)}_{\psi}}\\
				&\leq\EE{\psi\sim\Haar(n)}\TV{\Phi^{(t)}_\psi}{\psi^{\ot t}}+\delta\\
				&\leq \delta+\sqrt{\frac{t(t-1)}{2}\frac{d_A+1}{d_Ad_B+1}}.
		\end{split}\end{equation}
		where the inequalities follow from the triangle inequality, data processing, and Lemma \ref{lem:collision}, respectively. 
		
		Next, using Lemmas~\ref{lem:dephasing}, \ref{lem:reduction}, and ~\ref{lem:symmetric}, there exists a channel $\Lambda'$ such that
		\begin{equation}
			\EE{\phi\sim \mathrm{Haar}(n_A)}\left[\bra{\phi}^{\ot 2} \Lambda'(\phi)\ket{\phi}^{\ot 2}\right]
			\geq 1-\delta - \sqrt{\frac{t(t-1)}{2}\frac{d_A+1}{d_Ad_B+1}}.
		\end{equation}
		Finally, using Lemma \ref{lem:no_cloning}, we have
		\begin{equation}
			1-\delta - \sqrt{\frac{t(t-1)}{2}\frac{d_A+1}{d_Ad_B+1}}
			\leq \frac{2+4d_A}{(d_A+2)(d_A+1)},
		\end{equation}
		which implies that the average error $\delta$ must satisfy
		\begin{equation}
			\delta \geq 1-\sqrt{\frac{t(t-1)}{2}\frac{d_A+1}{d_Ad_B+1}} -\frac{2+4d_A}{(d_A+2)(d_A+1)}.
		\end{equation}
	\end{proof}
	
	\subsection{Further generalizations}\label{subsec:further_remarks}
	Finally, we remark that while we focus on ensembles induced by Haar-random purifications, our proof techniques and results readily extend to other ensemble distributions. We discuss two such generalizations below.
	
	\begin{enumerate}
		\item \emph{Distribution $p(\bm z)$ over projected states.} 
		
		In our main result, $p(\bm z)$ is induced by measuring a Haar-random state, but the same argument applies more generally to any distribution $\mu$ with small collision probability.
		
		Specifically, the distribution of $p(\bm z)$ enters our proof in only two places.
		(1) In Lemma~\ref{lem:reduction}, the constructed channel $\Lambda'$ simulates this distribution internally. Thus, one can replace the distribution of $p_0$ generated by $\Lambda'$ with $\mu$.
		(2) In Lemma~\ref{lem:collision}, if we sample a random distribution $p \sim \mu$ instead of the Haar-induced distribution $p_\psi$, then the collision term in Eq.~\eqref{eq:collision_union} is replaced by
		\begin{equation}
			C_{\mu} \coloneqq \EE{p\sim \mu}\sum_{\bmz} p(\bmz)^2.
		\end{equation}
		Therefore, the same proof gives
		\begin{equation}
			\mathbb{E}_{\psi}\ \TV{\Lambda(\psi^{\otimes t})}{\sigma_\psi^{(2)}} \ge 1 - \sqrt{\frac{t(t-1)}{2}C_{\mu}} - \frac{2+4d_A}{(d_A+2)(d_A+1)},
		\end{equation}
		where $\psi$ is an unknown random purification of the ensemble.
		
		\item
		\emph{Distribution of each projected state.}
		
		In our proof, each projected state $\psi_{\bm z}$ is Haar-randomly distributed. The same argument readily extends to more general distributions in which the projected states $\{\psi_{\bmz}\}$ are independently drawn from some ensemble $\mathcal{E}$, provided that $\mathcal{E}$ is invariant under some unitary $1$-design and that cloning an unknown state drawn from $\mathcal{E}$ is hard on average.
		
		More specifically, the distribution of $\psi_{\bm z}$ is used in two places.
		(1) In Lemma~\ref{lem:symmetric}, Eq.~(\ref{eq:twirling}) requires the distribution of $\psi_{\bm z}$ to be invariant under the twirling operation used in the proof. There, we apply Haar twirling so that the twirled states in Eq.~(\ref{eq:mu_i_twirling2}) and Eq.~(\ref{eq:trace_preseve}) become the maximally mixed state $I_A/d_A$. In fact, any unitary $1$-design already suffices for this purpose. Hence, as long as the distribution $\mathcal{E}$ of $\psi_{\bm z}$ is invariant under some unitary $1$-design, Lemma~\ref{lem:symmetric} continues to apply.
		(2) In Lemma~\ref{lem:no_cloning}, we require that cloning an unknown state $\psi_{\bm z}$ drawn from $\mathcal{E}$ is hard on average. Denote 
		\begin{equation}
			F_{\cE}^{\star} := \sup_{\Lambda} \EE{\phi\sim\mathcal E} \left[ \bra{\phi}^{\otimes 2} \Lambda(\phi) \ket{\phi}^{\otimes 2} \right].
		\end{equation}
		Combining these generalizations, we obtain
		\begin{equation}
			\mathbb{E}_{\psi}\ \TV{\Lambda(\psi^{\otimes t})}{\sigma_\psi^{(2)}} \ge 1 - \sqrt{\frac{t(t-1)}{2}C_{\mu}} - F_{\cE}^*,
		\end{equation}
		where $\psi$ is an unknown random purification of the ensemble $\{p(\bmz), \psi_{\bmz}\}$.
	\end{enumerate}
	
	\suppnote{Hardness of estimating nonlinear properties} \label{app:proof_universal}
	
	We present a detailed theorem that gives the explicit tradeoff between sample complexity and measurements for the estimation task, corresponding to the second part of Theorem~\ref{thm:summary}.
	
	\begin{theorem}[No-go on estimating nonlinear properties of quantum ensembles]\label{thm:universal_property}
		Let $0\leq\varepsilon< \tfrac{1}{20}$ and $0\leq\delta<\tfrac1{100}$ be fixed constants, and $n_A \ge 1$. Suppose there exists a quantum algorithm $\mathcal{A}$ such that given classical descriptions of any observable $O$ on $\mathcal{H}_A^{\otimes 2}$ with $\norm{O}_{\infty} \le 1$, and $t$ copies of an unknown input state $\ket{\psi}$, it estimates $o_{\psi}^{(2)}$ such that $\abs{\mathcal{A}\bigl(\ket{\psi}^{\otimes t}, O\bigr) - o_{\psi}^{(2)}} < \varepsilon$
		with probability at least $1-\delta$. Then, the number of copies required is $t = \Omega(2^{n_B / 4})$.
	\end{theorem}
	\begin{proof}
		We reduce the task of estimating $o^{(2)}_{\psi}$ to an ensemble–distinguishing problem between
		$\mathcal{E}_0$ and $\mathcal{E}_1 \coloneqq \mathrm{Haar}(n)$, where $\mathrm{Haar}(n)$ denotes the Haar distribution over pure $n$-qubit states.
		A referee samples a pure state $\ket{\phi}$ either from $\mathcal{E}_0$ or from $\mathcal{E}_1$, sends $\ket{\phi}^{\otimes t}$ to a prover, and asks the prover to identify which ensemble was used.
		
		Let $d \coloneqq 2^n$ and write $[d] = \{0, \ldots, d-1\}$. We identify $[d]$ with the set of $n$-bit strings $\{0,1\}^n = \{0,1\}^{n_A} \times \{0,1\}^{n_B}$, so that each element $\bm u \in [d]$ is viewed as a pair $(\bmx,\bmz)$ with $\bmx \in \{0,1\}^{n_A}$ labeling subsystem $A$ and $\bmz \in \{0,1\}^{n_B}$ labeling subsystem $B$. For a subset $S \subseteq [d]$ of size $m$, define the subset state~\cite{Tudor2023SubsetStates, Fernando2024SubsetStates}
		\begin{equation}
			\ket{S} := \frac{1}{\sqrt{|S|}} \sum_{\bm u \in S} \ket{\bm u}.   
		\end{equation}
		Let $\binom{[d]}{m}$ denote all size-$m$ subsets of $[d]$. The ensemble $\mathcal{E}_0$ is obtained by drawing $S \sim \binom{[d]}{m}$ uniformly at random and outputting $\ket{S}$.
		The following fact bounds the distance between the average $t$-copy states of $\mathcal{E}_0$ and $\mathcal{E}_1$.
		\begin{fact}[{\cite[Theorem 1.1]{Fernando2024SubsetStates}}]\label{fact:subset_state}
			For $t,m \in \mathbb{N}^+$,
			\begin{equation}
				D\left(
				\EE{S \sim \binom{[d]}{m}} \bigl[\ketbra{S}{S}^{\otimes t}\bigr],\;
				\EE{\phi \sim \mathrm{Haar}(n)} \bigl[\phi^{\otimes t}\bigr]
				\right)
				\le \mathcal{O}\left(\frac{t^2}{d} + \frac{t}{\sqrt{m}} + \frac{mt}{d}\right).
			\end{equation}
		\end{fact}
		
		By data processing, any quantum algorithm $\mathcal{B}$ acting on $t$ copies has a distinguishing advantage at most this trace distance. Hence
		\begin{equation}\label{eq:statistical_distance}
			\Bigl|
			\Pr_{S \sim \binom{[d]}{m}}\bigl[\mathcal{B}(\ket{S}^{\otimes t})=0\bigr]
			- \Pr_{\phi\sim \mathrm{Haar}(n)}\bigl[\mathcal{B}(\ket{\phi}^{\otimes t})=0\bigr]
			\Bigr|
			\le
			\mathcal{O}\left(\frac{t^2}{d} + \frac{t}{\sqrt{m}} + \frac{mt}{d}\right).
		\end{equation}
		Thus, no algorithm can reliably distinguish the ensembles unless $t$ is large.
		
		Now, assume we have an algorithm $\mathcal{A}$ that, given $t$ copies of any state $\ket{\psi}$ and an observable $O$, estimates $o^{(2)}_{\psi}$ to within error $0 \leq \varepsilon\le\frac{1}{20}$ with success probability $1-\delta\ge0.99$. We use $\mathcal{A}$ to construct a distinguisher $\mathcal{B}$ for the ensembles $\mathcal{E}_0$ and $\mathcal{E}_1$.
		
		The key is to find an observable $O$ for which the value of $o^{(2)}_{\psi}$ is characteristically different for the two ensembles. Let $d_A=2^{n_A}$ and $d_B=2^{n_B}$. We choose the observable 
		\begin{equation}\label{eq:measurable_observable}
			O \coloneqq \sum_{\bm{x}\in\{0,1\}^{n_A}} (\ketbra{\bm{x}})^{\otimes 2}
		\end{equation}
		on $\mathcal{H}_A^{\otimes 2}$. The following lemmas, whose proofs are deferred, show that the expected value of $o^{(2)}_{\psi}$ is well-separated for the two ensembles.
		\begin{lemma}\label{lem:o_2_subset}
			For $1 \le n_A \le n$ and $1 \le m \le d$,
			\begin{equation}
				\EE{\psi \sim \mathcal{E}_0}\left[o^{(2)}_{\psi}\right] \geq 1-\frac{m(m-1)}{2 d_B}.
			\end{equation}
		\end{lemma}
		\begin{lemma} \label{lem:o_2_Haar}
			For $1 \le n_A \le n$, 
			\begin{equation}
				\EE{\psi\sim \mathrm{Haar}(n)} [o^{(2)}_{\psi}] \leq \frac{2}{d_A+1}.
			\end{equation}
		\end{lemma}
		
		These lemmas show that for states from $\mathcal{E}_0$, $o^{(2)}_{\psi}$ is typically close to 1, while for states from $\mathcal{E}_1$, it is typically away from 1. Let $a = 1 - 5m(m-1)/d_B$ and $b = 8/(3(d_A+1))$. Then, by Markov's inequality, 
		\begin{equation}\label{eq:prob_bounds}
			\begin{split}
				\Pr_{\psi\sim \mathcal{E}_0}\left[o^{(2)}_{\psi}<a\right] &\le \frac{1}{10},  \\
				\Pr_{\psi\sim \mathcal{E}_1}\left[o^{(2)}_{\psi}>b\right] &\le \frac{3}{4}.
			\end{split}
		\end{equation}
		Our distinguisher $\mathcal{B}$ is defined as follows:
		\begin{enumerate}
			\item Set a decision threshold $\eta^* = (a+b)/2$.
			\item Given $t$ copies of an unknown state $\ket\psi$, run $\mathcal{A}(\ket{\psi}^{\otimes t}, O)$ and obtain an output $\eta$.
			\item If $\eta > \eta^*$, output $0$. Otherwise, output $1$.
		\end{enumerate}
		
		We choose parameters to ensure a large gap between $a$ and $b$. Set $m = c \cdot 2^{n_B/2}$ for a sufficiently small constant $c>0$. Then for large $n_B$, we have $a \approx 1$ and $b \le 8/9$, ensuring the gap $a-b$ is a constant larger than $2\varepsilon \le 0.01$.
		If the state $\ket{\psi}$ is from $\mathcal{E}_0$, then with probability at least $9/10$, $o^{(2)}_{\psi} \ge a$ (Eq.~\eqref{eq:prob_bounds}). In this case, since $\mathcal{A}$ succeeds with probability $1-\delta$, its output $\eta$ will satisfy $\eta \ge a-\varepsilon \ge \eta^*$. Thus, 
		\begin{equation}
			\Pr_{S \sim \binom{[d]}{m}}\bigl[\mathcal{B}(\ket{S}^{\otimes t})=0\bigr] \ge \frac{9}{10} (1-\delta)
		\end{equation}
		Similarly, if the state $\ket{\psi}$ is from $\mathcal{E}_1$, the probability that $\mathcal{B}$ incorrectly outputs 0 is bounded:
		\begin{equation}
			\Pr_{\phi\sim \mathrm{Haar}(n)}\bigl[\mathcal{B}(\ket{\phi}^{\otimes t})=0\bigr] \le \frac{3}{4} + \frac{\delta}{4}.
		\end{equation}
		Therefore,
		\begin{equation}
			\Bigl|
			\Pr_{S \sim \binom{[d]}{m}}\bigl[\mathcal{B}(\ket{S}^{\otimes t})=0\bigr]
			- \Pr_{\phi\sim \mathrm{Haar}(n)}\bigl[\mathcal{B}(\ket{\phi}^{\otimes t})=0\bigr]
			\Bigr|
			\ge \frac{9}{10} (1-\delta) - \frac{3}{4} - \frac{\delta}{4} \ge \frac{1}{10}.
		\end{equation}
		Combining with Eq.~\eqref{eq:statistical_distance}, 
		\begin{equation}\label{equation:precision}
			\frac{t^2}{d} + \frac{t}{\sqrt{m}} + \frac{mt}{d} = \Omega(1).
		\end{equation}
		which gives the claimed lower bound $t = \Omega(\sqrt{m}) = \Omega(2^{n_B / 4})$.
		
	\end{proof}
	
	\begin{proof}[Proof of Lemma \ref{lem:o_2_subset}]
		Let
		\begin{equation}
			\mathcal T:=\big\{S\in\binom{[d]}{m}:\ \text{the $B$-labels of the elements of $S$ are pairwise distinct}\big\}.
		\end{equation}
		Here, each element $\bm u \in [d]$ can be written as $\bm u=(\bmx,\bmz)$, where the $B$-label is the bit string $\bmz$ on subsystem $B$.
		Then
		\begin{equation}\begin{split}
				\EE{S \sim \binom{[d]}{m}}\left[o^{(2)}_{\ket{S}}\right] &= \EE{\substack{S\sim \binom{[d]}{m}\\ \psi_{\bmz} \sim \mathcal{E}_{S}}} [\tr(O \psi_{\bmz}^{\otimes 2})] \\
				&=\frac{1}{\binom{d}{m}}\sum_{S\in\binom{[d]}{m}}\EE{\psi_\bmz\sim {\cal E}_S}[\tr(O\psi_\bmz^{\otimes 2})]\\
				&\geq\frac{1}{\binom{d}{m}}\sum_{S\in\mathcal T}\EE{\psi_\bmz\sim {\cal E}_S}[\tr(O\psi_\bmz^{\otimes 2})].
		\end{split}\end{equation}
		
		For any $S\in\mathcal T$, each projected state $\psi_{\bm z}$ of $\ket{S}$ is a computational–basis product state $\ket{\bm x}$ on $A$, as for each $B$-label $\bm z$ there is at most one pair $(\bm x,\bm z)\in S$. Hence $\tr(O\psi_\bmz^{\otimes 2})=\tr(O(\ketbra{\bmx}{\bmx})^{\otimes 2})=1$ and 
		\begin{equation}
			\frac{1}{\binom{d}{m}}\sum_{S\in\mathcal T}\EE{\psi_\bmz\sim {\cal E}_S}[\tr(O\psi_\bmz^{\otimes 2})] = \Pr_{S\sim\binom{[d]}{m}}[S\in\mathcal T].
		\end{equation}
		
		It remains to lower bound $\Pr[S\in\mathcal T]$, i.e., the probability of no collisions among the $B$-labels. Label the elements of $S$ as $(\bm x_i,\bm z_i)_{i=1}^m$. For any fixed pair $1\le i<j\le m$, we have $\Pr[\bmz_i=\bmz_j]\leq 1/d_B$ by uniformity over $[d]$. Then, by the union bound,
		\begin{equation}\begin{split}
				\Pr_{S\sim \binom{[d]}{m}}[S\in{\cal T}] \geq& 1-\sum_{1<i<j\leq m}\Pr[\bmz_i=\bmz_j]\\
				\geq&1- \frac{m(m-1)}{2d_B},
		\end{split}\end{equation}
		which completes the proof. 
	\end{proof}
	
	\begin{proof}[Proof of Lemma \ref{lem:o_2_Haar}]
		By unitary invariance, if $\psi_{AB}$ is Haar–random on $\mathcal{H}_A\otimes\mathcal{H}_B$ then, for any fixed computational–basis outcome $\bm z$ on $B$, the post–measurement state $\psi_{\bm z}$ is Haar–distributed on $\mathcal{H}_A$. Hence 
		\begin{equation}\begin{split}
				\EE{\substack{\psi\sim \mathrm{Haar}(n)\\ \psi_{\bmz} \sim \mathcal{E}_{\psi}}} [\tr(O\psi_{\bmz}^{\otimes 2})]
				&=\EE{\psi \sim \mathrm{Haar}(n_A)} [\sum_{\bmx\in\{0,1\}^{n_A}}\bra{\bmx}^{\otimes 2}\psi^{\otimes 2}\ket{\bmx}^{\otimes 2}]\\
				&=d_A\EE{\psi \sim \mathrm{Haar}(n_A)} [\bra{0}^{\otimes 2}\psi^{\otimes 2}\ket{0}^{\otimes 2}]\\
				&=\frac{1}{d_A+1}[\bra{0}^{\otimes 2} (I+\mathbb{S})\ket{0}^{\otimes 2}]\\
				&=\frac{2}{d_A+1}.
		\end{split}\end{equation}
	\end{proof}
	
	\subsection{Hardness of observing measurement-induced quantum phenomena}
	
	Having established the hardness of estimating second-order quantities, we now show that this hardness extends to any quantity that distinguishes ensembles of computational-basis states from projected ensembles of Haar-random states.
	
	\begin{corollary}[Hardness of observing measurement-induced quantum phenomena, formal version of Corollary~1]
		Fix $n_A \ge 1$, and let $f$ be a $[0,1]$-valued function on quantum ensembles over $\mathcal{H}_A$, representing some property of interest. Suppose that there exists a constant $\Delta>0$ such that
		\begin{equation}
			\abs{f(\tilde \cE)-\EE{\psi\gets \Haar(n)}f(\cE_{\psi,A})}>\Delta
		\end{equation}
		for every ensemble $\tilde{\cE}$ supported on computational-basis states.
		Then there exists some constants $\delta,\varepsilon>0$ such that, for any quantum algorithm $\mathcal{A}_f$ that estimates $f(\cE)$ for $t$ copies of a purification $\psi$ of $\mathcal E$, and satisfies
		\begin{equation}
			\Pr_{\mathcal A_f}\left[\abs{\mathcal{A}_f\bigl(\psi^{\otimes t}\bigr) - f(\cE)} < \varepsilon\right]>1-\delta,
		\end{equation}
		the required number of copies must satisfy $t = \Omega(2^{n_B / 4})$.
	\end{corollary}
	\begin{proof}
		
		The proof is similar to that of Theorem~\ref{thm:universal_property}.
		By assumption, $f$ has a constant gap $\Delta$ between the projected ensembles generated by $\Haar(n)$ and any ensemble supported on computational-basis states. By the collision-probability analysis in the proof of Lemma~\ref{lem:o_2_subset}, the projected ensemble induced by $\cE_0$ contains non-computational-basis states with probability at most $\frac{m(m-1)}{d_B}$. 
		Therefore, 
		\begin{equation}
			\abs{\EE{\psi\gets \cE_0}f(\cE_{\psi,A})-\EE{\psi\gets \cE_1}f(\cE_{\psi,A})}\geq \Delta-\frac{m(m-1)}{d_B}.
		\end{equation}
		Setting $m=O(2^{n_B/2})$, this gap remains bounded below by a positive constant.
		The rest then follows exactly as the proof of Theorem~\ref{thm:universal_property}. 
	\end{proof}
	
	This corollary directly implies the worst-case hardness of observing measurement-induced entanglement, magic, complexity, and deep thermalization.
	
	For measurement-induced entanglement, define $f$ as the fraction of entangled states in a given ensemble:
	\begin{equation}
		f(\mathcal E)=\Pr_{\psi\sim \mathcal E}[\psi \text{ is entangled}].
	\end{equation}
	Since every computational-basis state is separable, we have $f(\tilde{\mathcal E})=0$ for any ensemble $\tilde{\mathcal E}$ supported on computational-basis states.
	On the other hand, for the projected ensembles of Haar-random states, the expectation of $f$ is exactly the fraction of entangled states among Haar-random $n_A$-qubit pure states. Since the set of separable states has Haar measure zero, this expectation equals $1$. Thus, the gap is $\Delta=1$, and it follows that the worst-case sample complexity for estimating $f$ is exponential in $n_B$.
	This argument also extends to more robust entanglement-related quantities, such as subsystem purity or entanglement entropy (after adjusting the range of $f$ accordingly), since Haar-random states are highly entangled with high probability~\cite{Liu2018EntanglementDesign}.
	
	Likewise, by taking $f$ to be the fraction of magic states, or the fraction of states whose circuit complexity exceeds a fixed threshold, one obtains the worst-case hardness of observing measurement-induced magic and circuit complexity.
	
	For deep thermalization, let $f_{k,\epsilon}(\cE)$ indicate whether $\cE$ forms an $\epsilon$-approximate $k$-design, for fixed $k=\cO(1)$ and constant $\epsilon>0$. When $n_B=\Omega(n_A)$, the projected ensemble of a Haar-random state forms such an approximate design with high probability~\cite{Cotler2023EmergentDesign}, whereas any ensemble of computational-basis states cannot form even an $\epsilon$-approximate $2$-design for some constant $\epsilon$. This again gives a constant gap $\Delta>0$, implying the hardness of verifying deep thermalization in the regime $n_B=\Omega(n_A)$, which is a standard regime where deep thermalization is expected to emerge~\cite{Cotler2023EmergentDesign}.

	\suppnote{Sample-efficient algorithm on bounded-complexity ensembles}\label{app:sample-efficient-physical}
	
	Here, we prove Theorem~\ref{thm:sample-efficient-learning}, which shows that it is possible to prepare the state $\sigma_{\psi}^{(k)}$ in a sample-efficient manner when the input state $\ket{\psi}$ is generated by a polynomial-size quantum circuit $U$. 
	
	The proof proceeds by explicitly constructing an algorithm that achieves this performance bound. 
	The algorithm has two stages: it first learns the input state via a tomography procedure tailored to states with bounded gate complexity~\cite{Zhao2024BoundedGateComplexity}. Then it uses the resulting classical description to efficiently prepare a quantum state approximating $\sigma_{\psi}^{(k)}$. 
	The tomography algorithm provides the following performance guarantee:
	\begin{fact}[{\cite[Theorem~1]{Zhao2024BoundedGateComplexity}}]\label{fact:sample-efficient}
		Given $0 < \varepsilon < 1$ and $N$ copies of an $n$-qubit input state $\ket{\psi} = U\ket{0}^{\otimes n}$, where the unitary $U$ consists of $G$ two-qubit gates, one can obtain a classical description of a pure state $\ket{\phi}$ such that $D(\psi, \phi) \le \varepsilon$ provided that $N = \widetilde{\Theta}\left(\frac{G}{\varepsilon^2}\right)$. Here, $\widetilde{\Theta}(\cdot)$ hides logarithmic factors.
	\end{fact}
	
	We now show that if the learned state $\ket{\phi}$ is close to the target state $\ket{\psi}$, then their corresponding higher-order projected states $\sigma^{(k)}_{\phi}$ and $\sigma^{(k)}_{\psi}$ are also close.
	\begin{lemma}\label{lem:bound-high-order-distance}
		Let $k \ge 1$. For two pure states $\ket{\phi}$ and $\ket{\psi}$ satisfying $D(\phi,\psi) \le \varepsilon$, we have $D(\sigma^{(k)}_{\psi}, \sigma_{\phi}^{(k)}) \le (2k+1) \varepsilon$.
	\end{lemma}
	
	\begin{proof}
		First, note that $\sigma_{\psi}^{(1)} = \Delta_B(\psi)$, where $\Delta_B$ denotes the dephasing channel acting on subsystem $B$. By the data-processing inequality, 
		\begin{equation}
			D(\sigma^{(1)}_{\psi}, \sigma_{\phi}^{(1)}) \le \varepsilon.
		\end{equation} 
		Consequently,
		\begin{equation}
			\mathrm{TV}(p_{\psi}, p_{\phi}) = D(\tr_A(\sigma_{\psi}^{(1)}), \tr_A(\sigma_{\phi}^{(1)})) \le  D(\sigma_{\psi}^{(1)}, \sigma_{\phi}^{(1)}) \le \varepsilon.
		\end{equation}
		where $\mathrm{TV}$ denotes the total variation distance. 
		
		Define the state $\tilde{\sigma}_{\phi}^{(k)} \coloneqq \sum_{\bmz}p_{\psi}(\bmz) \phi_{\bmz}^{\otimes k} \otimes \ketbra{\bmz}$, which replaces the distribution $p_{\phi}$ in $\sigma_{\phi}^{(k)}$ by $p_{\psi}$. 
		Then, 
		\begin{equation}
			D(\tilde{\sigma}_{\phi}^{(k)}, \sigma_{\phi}^{(k)}) = \mathrm{TV}(p_{\psi}, p_{\phi}) \le \varepsilon.
		\end{equation}
		It follows that
		\begin{equation}
			D(\tilde{\sigma}_{\phi}^{(1)}, \sigma_{\psi}^{(1)}) \le  D(\tilde{\sigma}_{\phi}^{(1)}, \sigma_{\phi}^{(1)}) +  D(\sigma_{\phi}^{(1)}, \sigma_{\psi}^{(1)}) \le 2 \varepsilon. 
		\end{equation}
		\begin{equation}
			D(\tilde{\sigma}_{\phi}^{(k)},\sigma_{\psi}^{(k)}) = \sum_{\bmz} p_{\psi}(\bmz) D(\phi_{\bmz}^{\otimes k}, \psi_{\bmz}^{\otimes k}) \le k\sum_{\bmz} p_{\psi}(\bmz) D(\phi_{\bmz}, \psi_{\bmz}) = k D(\tilde{\sigma}_{\phi}^{(1)}, \sigma_{\psi}^{(1)}) \le 2k\varepsilon,
		\end{equation}
		where the first inequality comes from
		\begin{equation}
			D(\phi_{\bmz}^{\otimes k}, \psi_{\bmz}^{\otimes k}) \le \sum_{i=1}^{k} D(\phi_{\bmz}^{\otimes k-i+1} \otimes \psi_{\bmz}^{i-1}, \phi^{k-i}_{\bmz} \otimes \psi_{\bmz}^{\otimes i}) = \sum_{i=1}^k D(\phi_{\bmz}, \psi_{\bmz}) = kD(\phi_{\bmz}, \psi_{\bmz}).
		\end{equation}
		Finally,
		\begin{equation}
			\begin{split}
				D(\sigma^{(k)}_{\psi}, \sigma_{\phi}^{(k)}) \le   D(\tilde{\sigma}_{\phi}^{(k)}, \sigma_{\phi}^{(k)}) + D(\tilde{\sigma}_{\phi}^{(k)},\sigma_{\psi}^{(k)}) \le (2k + 1) \varepsilon.
			\end{split}
		\end{equation}
	\end{proof}

	\begin{proof}[Proof of Theorem~\ref{thm:sample-efficient-learning}]
		We begin by setting $\varepsilon' = \frac{\varepsilon}{2k + 1}$ and applying the tomography protocol described in Fact~\ref{fact:sample-efficient} to obtain a classical description of a pure state $\phi$ satisfying $D(\phi, \psi) \le \varepsilon'$ using $N = \widetilde{\Theta}\left(\frac{G}{\varepsilon'^2}\right) = \widetilde{\Theta}\left(\frac{k^2G}{\varepsilon^2}\right)$ copies of $\ket{\psi}$. 
		This procedure also yields a classical description of the corresponding state $\sigma_{\phi}^{(k)}$. 
		Then, by Lemma~\ref{lem:bound-high-order-distance}, we have $D(\sigma_{\phi}^{(k)}, \sigma_{\psi}^{(k)}) \le (2k+1)\varepsilon' = \varepsilon$, which completes the proof.
	\end{proof}

	\suppnote{Hardness of estimating nonlinear properties with prior knowledge}
	\label{app:no-go-prior} 
	We now present detailed statements and proofs of the two parts of Theorem~\ref{thm:no-go-prior}. We begin by introducing the definition of QPRF and QPRP, which are function families that cannot be distinguished from truly random ones by polynomial-time quantum algorithms, as well as the LWE and DLWE problems, which are standard computational assumptions in post-quantum cryptography.
	
	\begin{definition}[Quantum-secure pseudo-random functions, QPRF \cite{Zhandry2021QuantumRandomFunctions}]
		Let $\lambda$ be the security parameter, $n,l =\cO(\poly(\lambda))$ be the input length and output length, respectively. Let $\mathcal{K}_{\lambda}$ be key space.
		A family of efficient classical functions $\{f_{k,n}: \{0,1\}^n \rightarrow \{0,1\}^l \}_{\lambda, k \in \mathcal{K_{\lambda}}}$ is a family of quantum-secure pseudo-random functions if for any $\poly(\lambda)$-time quantum algorithm $\mathcal{A}$ with superposition query to the oracle, we have that
		\begin{equation}
			\abs{\Pr_{k} [\mathcal{A}^{f_{k,n}}(1^\lambda) = 1] -  \Pr_{r}[\mathcal{A}^{r}(1^\lambda) = 1]} \le \negl(\lambda)
		\end{equation}
		where $r:\{0,1\}^n\to\{0,1\}^l$ is a truly random function.
	\end{definition}
	
	\begin{definition}[Quantum-secure pseudo-random permutations, QPRP \cite{Zhandry2025notequantumsecure}]
		Let $\lambda$ be the security parameter, $n=\cO(\poly(\lambda))$ be the input length and $\mathcal{K}_{\lambda}$ be key space.
		A family of efficient classical permutations $\{f_{k,n}: \{0,1\}^n \rightarrow \{0,1\}^n \}_{\lambda, k \in \mathcal{K_{\lambda}}}$ is a family of quantum-secure pseudo-random permutations if for any $\poly(\lambda)$-time quantum algorithm $\mathcal{A}$ with superposition query to the oracle, we have that
		\begin{equation}
			\abs{\Pr_{k} [\mathcal{A}^{f_{k,n},f_{k,n}^{-1}}(1^\lambda) = 1] -  \Pr_{r}[\mathcal{A}^{r,r^{-1}}(1^\lambda) = 1]} \le \negl(\lambda)
		\end{equation}
		where $r:\{0,1\}^n\to\{0,1\}^n$ is a truly random permutation (i.e. uniformly sampled from all $n$-bit permutations).
	\end{definition}
	
	We denote $\cU(S)$ as the uniform distribution over a set $S$. Let $\mathbb Z_q$ be the ring of integers modulo $q$. The addition and multiplication of two elements over $\mathbb Z_q$ are under modulo $q$. Denote $\cD_{\mathbb Z_q,\beta^2}$ as the discrete Gaussian distribution with `variance' $\beta^2$ over $\mathbb Z_q$. In particular, for every $x\in \mathbb Z_q$,
	\begin{equation}
		\cD_{\mathbb Z_q,\beta^2}(x)=\frac{e^{-\frac{ |x|^2}{2\beta^2}}}{\sum_{z\in \mathbb Z_q}e^{-\frac{|z|^2}{2\beta^2}}}
	\end{equation}
	where $|x|:=\min(x,q-x)$. For a vector $\bmx=(x_1,x_2,\cdots,x_m)\in \mathbb Z_q^m$, define the infinity norm $\|\bmx\|_\infty$ of $\bmx$ as
	\begin{equation}
		\|\bmx\|_\infty:=\max_{i=1}^m|x_i|.
	\end{equation}
	
	\begin{definition}[Learning with errors problem (LWE) \cite{regev2009lattices}]
		In learning with errors problem with parameter $n,m(n),q(n),\beta(n)$ (denoted as $\LWE_{n,m,q,\beta}$), there is a secret $\bm s\gets \cU(\mathbb Z_q^n)$. 
		We are given a uniform random matrix $A\gets \cU(\mathbb Z^{n\times m}_q)$ and a vector $\bm y=A^T\bms+\bm e$ with $\bm e\gets \cD_{\mathbb Z_q,\beta^2}^m$. The goal is to recover the secret $\bms$.
	\end{definition}
	
	With these definitions in place, we now state the detailed versions of Theorem~\ref{thm:no-go-prior}.
	
	\begin{theorem}\label{thm:no-go-physical}
		Let $0\leq\varepsilon< \tfrac{1}{10}$ and $0\leq\delta<\tfrac1{3}$ be fixed constants, and suppose $n_A \ge 1, n_B = \omega(\log n)$. 
		Assume the existence of QPRF and QPRP.  
		Suppose there exists an algorithm $\mathcal{A}$ such that given classical descriptions of an efficiently measurable observable $O$, and access to copies of an efficiently preparable unknown input state $\ket{\psi}$, it estimates $o_{\psi}^{(2)}$ such that
		\begin{equation}
			\abs{\mathcal{A}\bigl(\ket{\psi}^{\otimes t}, O\bigr) - o_{\psi}^{(2)}} < \varepsilon
		\end{equation}
		with probability at least $1-\delta$. Then, $\mathcal{A}$ requires $n^{\omega(1)}$ computational time. 
	\end{theorem}
	
	\begin{theorem}\label{thm:no-go-physical-prior}
		Let $0\leq\varepsilon< \tfrac{1}{10}$ and $0\leq\delta<\tfrac1{3}$ be fixed constants, and suppose $n_A,n_B = \Omega(n^c)$ for some constant $0 < c \le 1$. 
		Assuming the $\LWE_{n,m,q,\beta}$ cannot be solved by any polynomial-time quantum algorithm for some $m=2n\log q,q=\poly(n),\beta\leq\frac{q}{32n^{1.5}\log q}$ with prime $q$.
		Suppose there exists an algorithm $\mathcal{A}$ such that given classical descriptions of a polynomial-sized quantum circuit $C$ for preparing a state $\ket\psi = C\ket{0}^{\otimes n}$, and an efficiently measurable observable $O$, it estimates $o_{\psi}^{(2)}$ such that
		\begin{equation}
			\abs{\mathcal{A}\bigl(C, O\bigr) - o_{\psi}^{(2)}} < \varepsilon
		\end{equation}
		with probability at least $1-\delta$. Then, $\mathcal{A}$ requires $n^{\omega(1)}$ computational time. 
	\end{theorem}
	
	The parameters we choose are standard for LWE. Regev’s quantum reduction~\cite{regev2009lattices} implies that this LWE parameter regime is at least as hard as $\tilde \Theta(n^{2.5})$-approximation of the shortest vector problem, which does not have a sub-exponential time algorithm currently.
	
	\subsection{Proof of Theorem \ref{thm:no-go-physical}} \label{app:proof_physical_state}
	
	Now, we construct two ensembles of efficiently preparable states that are computationally indistinguishable by every polynomial-time algorithm, and show that they can be distinguished efficiently if $o^{(2)}$ can be estimated efficiently.
	Since each copy can be prepared in polynomial time, one may equivalently regard the algorithm as having unit-time access to each copy, which does not affect the resulting super-polynomial lower bound.
	This leads to the conclusion that estimating $o^{(2)}$ requires super-polynomial time.
	
	Fix a parameter $m = n^{\omega(1)}$ such that
	\begin{equation}\label{eq:choose_nB}
		n_B = \Omega\left\{\log [m^2 n^{\omega(1)}] \right\}.
	\end{equation}
	This can be done since $n_B=\omega(\log n)$.
	Let $f$ be drawn from a family of QPRPs on $n(\lambda) = n_B$, and let $g$ be drawn from a family of QPRFs with $n(\lambda) = n_B$ and $l(\lambda) = n_A$.  Consider the following two states:
	\begin{equation}
		\ket{\psi_{f,g}}  = \frac{1}{\sqrt{m}} \sum_{x=0}^{m-1}  \ket{g(x)}_A \ket{f(x)}_B
	\end{equation}
	and
	\begin{equation}
		\ket{\phi_{f,g}}= (H^{\otimes n_A}\otimes I^{\otimes n_B}) \ket{\psi_{f,g}} = \frac{1}{\sqrt{m}} \sum_{x=0}^{m-1} \left(H^{\otimes n_A} \ket{g(x)}_A\right) \ket{f(x)}_B,
	\end{equation}
	where $H$ denotes the Hadamard gate.
	The two ensembles are given by $\mathcal{E}_{\psi} \coloneqq \{\ket{\psi_{f,g}}\}$ and $\mathcal{E}_{\phi} \coloneqq \{\ket{\phi_{f,g}}\}$.
	
	Both kinds of states are efficiently preparable. Taking $\ket{\psi_{f,g}}$ as an example, one can first prepare the superposition $1/\sqrt{m}\sum_{x=0}^{m-1}\ket{x}$, and then compute $g$ to prepare $1/\sqrt{m}\sum_{x=0}^{m-1}\ket{g(x)}\ket{x}$. Applying the invertible function $f$ to the second register then obtains $\ket{\psi_{f,g}}$. The state $\ket{\phi_{f,g}}$ can be prepared by applying $H^{\otimes n_A}$ to $\ket{\psi_{f,g}}$.
	
	Now, we show that the two ensembles have distinct $o^{(2)}$ values for a certain efficiently measurable observable. The observable is chosen as in Eq.~\eqref{eq:measurable_observable}. 
	This observable is efficiently measurable because the underlying function $f(\bmx_1, \bmx_2) = \mathbf{1}[\bmx_1 = \bmx_2]$ can be computed efficiently.
	
	\begin{lemma}\label{lem:physical_o_value}
		Given the ensembles $\mathcal{E}_{\psi}$ and $\mathcal{E}_{\phi}$ defined above, we have
		\begin{equation}
			\begin{split}
				o_{\psi}^{(2)} = 1, \quad &\forall \psi \in \mathcal{E}_{\psi}, \\
				o_{\phi}^{(2)} = \frac{1}{d_A}, \quad &\forall \phi \in \mathcal{E}_{\phi}.
			\end{split}
		\end{equation}
	\end{lemma}

	\begin{proof}
		For any state $\ket{\psi} \in \mathcal{E}_{\psi}$ and $\phi = H^{\otimes n_A} \ket{\psi}$, 
		after measuring the $B$ subsystem, each projected state $\psi_{\bm{z}}$ has the form $\psi_{\bm{z}} = \ket{x_{\bm{z}}}$ for a certain $x_{\bm{z}}$, as there is no collision on the $B$ labels due to $f$ being a permutation. 
		Meanwhile, each state $\phi_{\bm{z}}$ has the form $\phi_{\bm{z}} = H^{\otimes n_A} \ket{x_{\bm{z}}}_A$. The former has $\tr(\psi_{\bm{z}}^{\otimes 2} O) = 1$, and the latter has $\tr(\phi_{\bm{z}}^{\otimes 2} O) = \frac{1}{d_A}$. Therefore, every state from $\mathcal{E}_{\psi}$ has $o_{\psi}^{(2)} = 1$, and every state from $\mathcal{E}_{\phi}$ has $o_{\phi}^{(2)} = \frac{1}{d_A}$.
	\end{proof}
	
	Next, we prove that these two ensembles are computationally indistinguishable.
	
	\begin{lemma}
		\label{lem:indistinguishable}
		For every $t=\poly(n)$ and any $\poly(n)$-time quantum algorithm $\mathcal A$,
		\begin{equation}
			\left|\Pr_{\psi\sim\cE_\psi}[\mathcal A(\psi^{\otimes t})=1]-\Pr_{\phi\sim\cE_\phi}[\mathcal A(\phi^{\otimes t})=1]\right|\leq \negl(n).
		\end{equation}
	\end{lemma}
	\begin{proof}
		Select $\lambda=n$ as the security parameter.
		We first define the following ensembles. Let
		\begin{itemize}
			\item $\cE_{\psi,0} \coloneqq \cE_\psi$;
			\item $\cE_{\psi,1} \coloneqq \{\ket{\psi_{r,g}}\}$ where $r$ is a truly random permutation and $g$ is sampled from the family of QPRFs as in $\cE_\psi$.
			\item $\cE_{\psi,2} \coloneqq \{\ket{\psi_{r_1,r_2}}\}$ where $r_1$ is a truly random permutation and $r_2$ is a truly random function;
			\item $\cE_{\psi,3} \coloneqq \{\ket{S}\}$ be the ensemble of subset state where $S\sim \binom{[d]}{m}$. 
			\item $\cE_{\psi,4} \coloneqq \mathrm{Haar}(n)$ be the ensemble of Haar-random states;
			\item $\cE_{\phi,i}=\{(H_A^{\otimes n_A}\otimes I_B)\ket{\psi}: \psi \sim \mathcal{E}_{\psi,i}\}$ for each $i\in[0,4]$.
		\end{itemize}
		
		By the unitary invariance of Haar measure, $\cE_{\psi,4}$ and $\cE_{\phi,4}$ are identical. To show that $\cE_\psi=\cE_{\psi,0}$ and $\cE_\phi=\cE_{\phi,0}$ are computationally indistinguishable, we only need to show that every $\cE_{\psi,i}$ and $\cE_{\psi,i+1}$ are computationally indistinguishable given $t$ samples and apply a hybrid argument. 
		
		We first show that $\cE_{\psi,0}$ and $\cE_{\psi,1}$ are computationally indistinguishable for any $\poly(n)$ time adversary. Suppose, for contradiction, that there exists a $\poly(n)$-time adversary $\mathcal{B}$ that can distinguish them with non-negligible advantage by taking $t$ copies. Then, we can construct a polynomial-time algorithm $\mathcal{A}$ to break the pseudo-random property of $f$ as follows:
		\begin{equation}
			\mathcal{A}^{f,f^{-1}}(1^\lambda) = \mathcal{B}(\ket{\psi_{f,g}}^{\otimes t}),
		\end{equation}
		where $g$ is sampled from the QPRF ensemble. This contradicts the pseudo-random property of $f$, and thus $\mathcal{E}_{\psi,0}$ and $\mathcal{E}_{\psi,1}$ are computationally indistinguishable.
		
		To show that $\cE_{\psi,1}$ and $\cE_{\psi,2}$ are computationally indistinguishable, we again assume the existence of a $\poly(n)$-time distinguisher $\cB$, and construct $\cA$ as follows:
		\begin{equation}
			\cA^{g}(1^\lambda)= \mathcal B(\ket{\psi_{f,g}}^{\otimes t}),
		\end{equation}
		where $f$ is sampled from the QPRP ensemble.
		By the pseudo-random property of $f$, the output distribution will be close to $\mathcal B(\ket{\psi_{r_1,g}}^{\otimes t})$
		where $r_1$ is a truly random permutation. Therefore, $\mathcal{A}^{g}(1^\lambda)$ and $\mathcal{A}^{r_2}(1^\lambda)$ are distinguishable. This contradicts the pseudo-random property of $g$, and thus $\mathcal{E}_{\psi,1}$ and $\mathcal{E}_{\psi,2}$ are computationally indistinguishable.
		
		For $\mathcal{E}_{\psi,2}$ and $\mathcal{E}_{\psi,3}$, recall that in the proof of Lemma~\ref{lem:o_2_subset} we defined $\mathcal{T} \subseteq \binom{[d]}{m}$ as the collection of $m$-element subsets whose elements carry pairwise distinct $B$-labels. The ensemble $\mathcal{E}_{\psi,2}$ is exactly the subset states uniformly distributed over $\mathcal{T}$. 
		The statistical distance between $\mathcal{E}_{\psi,2}$ and $\mathcal{E}_{\psi,3}$ is at most the total variation distance between the uniform distribution over $\binom{[d]}{m}$ and over $\mathcal{T}$, which is $1 - |\mathcal{T}|/\binom{d}{m}$. In the proof of Lemma~\ref{lem:o_2_subset}, we have shown that $1 - |\mathcal{T}|/\binom{d}{m} \leq \frac{m(m-1)}{2d_B}$, which is negligible for the choice of $m$ in Eq.~\eqref{eq:choose_nB}. Therefore, $\mathcal{E}_{\psi,2}$ and $\mathcal{E}_{\psi,3}$ are statistically indistinguishable and hence computationally indistinguishable.

		By Fact \ref{fact:subset_state}, $\cE_{\psi,3}$ and $\cE_{\psi,4}$ are statistically indistinguishable for $t=\poly(n)$, and hence they are computationally indistinguishable.

		We have established that the ensembles $\cE_{\psi,i}$ and $\cE_{\psi,i+1}$ are computationally indistinguishable for each $i$. This implies that each $\cE_{\phi,i}$ and $\cE_{\phi,i+1}$ also are computationally indistinguishable. Now we apply a standard hybrid argument. Consider any $\poly(n)$-time adversary $\mathcal A$. The advantage of $\mathcal{A}$ distinguishing the initial hybrid $\mathcal{E}_\psi = \mathcal{E}_{\psi,0}$ and the final hybrid $\mathcal{E}_\phi = \mathcal{E}_{\phi,0}$ is bounded by 
		\begin{equation}\begin{split}
				\left|\Pr_{\psi\sim\cE_\psi}[\mathcal A(\psi^{\otimes t})=1]-\Pr_{\phi\sim\cE_\phi}[\mathcal A(\phi^{\otimes t})=1]\right|
				&\leq \sum_{i=0}^3\left|\Pr_{\psi \sim\cE_{\psi,i}}[\mathcal A(\psi^{\otimes t})=1]-\Pr_{\psi\sim\cE_{\psi,i+1}}[\mathcal A(\psi^{\otimes t})=1]\right|\\
				&\quad +\sum_{i=0}^3\left|\Pr_{\phi\sim\cE_{\phi,i+1}}[\mathcal A(\phi^{\otimes t})=1]-\Pr_{\phi\sim\cE_{\phi,i}}[\mathcal A(\phi^{\otimes t})=1]\right|\\
				&\leq \negl(\lambda).
		\end{split}\end{equation}
		This concludes the proof.
	\end{proof}
	
	\begin{proof}[Proof of Theorem \ref{thm:no-go-physical}]
		We construct an algorithm $\mathcal{B}$ to distinguish $\cE_\psi$ and $\cE_\phi$ given $t$ copies as follows. It first obtains a number $a = \mathcal{A}(\ket{\psi}^{\otimes t}, O)$ for $O$ given in Eq.~\eqref{eq:measurable_observable}. Then, it outputs $1$ if $a \geq 3/4$, and $0$ otherwise. 
		
		By Lemma \ref{lem:physical_o_value} and the definition of $\mathcal{A}$ in Theorem \ref{thm:no-go-physical}, we have that 
		\begin{equation}
			\begin{split}
				\Pr_{\psi\sim\cE_\psi}[\mathcal B(\psi^{\otimes t})=1] \ge 1-\delta, \\
				\Pr_{\phi\sim\cE_\phi}[\mathcal B(\phi^{\otimes t})=0] \ge 1-\delta.
			\end{split}
		\end{equation}
		This gives that 
		\begin{equation}
			\left|\Pr_{\psi\sim\cE_\psi}[\mathcal B(\psi^{\otimes t})=1] -  \Pr_{\phi\sim\cE_\phi}[\mathcal B(\phi^{\otimes t})=1] \right| \ge 1 - 2\delta,
		\end{equation}
		By Lemma \ref{lem:indistinguishable}, $\mathcal{B}$ requires $\omega(\poly(n))$ computational time, and therefore so does $\mathcal{A}$.
	\end{proof}
	
	\subsection{Proof of Theorem \ref{thm:no-go-physical-prior}} \label{app:no-go-physical-prior}
	We first introduce the decisional version of the LWE problem.
	\begin{definition}[Decisional learning with errors problem (DLWE) \cite{regev2009lattices}]
		In the decisional learning with errors problem with parameter $n,m(n),q(n),\beta(n)$ (denoted as $\DLWE_{n,m,q,\beta}$), we are given a uniform random matrix $A\gets \cU(\mathbb Z^{n\times m}_q)$ and a vector $\bm y\in\mathbb Z_q^m$. We are asked to distinguish which distribution $y$ is drawn from:
		\begin{itemize}
			\item LWE distribution: $\bm y=A^T\bms+\bm e$ for $\bms\gets \cU(\mathbb Z_q^n)$ and $\bm e\gets \cD_{\mathbb Z_q,\beta^2}^m$;
			\item Uniform distribution: $\bm y\gets \cU(\mathbb Z_q^m)$.
		\end{itemize}
	\end{definition}
	The following fact shows that DLWE is as hard as LWE.
	\begin{fact}[Decision to search reduction for LWE {\cite[Lemma 4.2]{regev2009lattices}}]
		\label{fact:desicion-to-search}
		For $2\leq q\leq \poly(n)$ which is a prime, if there exists a classical/quantum $T(n)$-time algorithm $\cA$ to solve $\DLWE_{n,m,q,\beta}$ with failure probability exponentially small in $n$, then there exists a $T(n)\cdot \poly(n)$-time classical/quantum algorithm to solve $\LWE_{n,m,q,\beta}$ with failure probability exponentially small in $n$.
	\end{fact}
	We now show that an efficient algorithm to estimate $o^{(2)}$ can be used to solve DLWE with some standard parameters.
	To avoid confusion, we denote $n_{AB}$ as the size of the system to estimate $o^{(2)}$, and $n$ as the size parameter for $\DLWE$ in the following.
	\begin{lemma}
		\label{lemma:reduce-lwe}
		Let $0\leq\varepsilon< \tfrac{1}{10}$ and $0\leq\delta<\tfrac1{3}$ be fixed constants, and let $n_A,n_B\geq S(n_{AB})$ for some non-decreasing function $S$. Assume there exists an efficient algorithm $\cA$ such that given classical descriptions of a polynomial-sized quantum circuit $C$ for preparing a state $\ket\psi = C\ket{0}^{\otimes n_{AB}}$, and an efficiently measurable observable $O$, it estimates $o_{\psi}^{(2)}$ such that
		\begin{equation}
			\abs{\mathcal{A}\bigl(C, O\bigr) - o_{\psi}^{(2)}} < \varepsilon
		\end{equation}
		with probability at least $1-\delta$, then there exists a $\poly(S^{-1}(\cO(n\log^2 n)))$-time algorithm for $\DLWE_{n,m,q,\beta}$ with $m=2n\log q,q=\poly(n),\beta\leq\frac{q}{32n^{1.5}\log q}$ for prime $q$. Here $S^{-1}$ is the inverse function of $S$.
	\end{lemma}
	\begin{proof}
		Suppose we are given a matrix $A$ and a vector $\bm y$.
		Set $\eta=\lceil 2\beta m\sqrt n+1\rceil$. For $b\in\{0,1\},\bmx\in\mathbb Z_q^n,\bm u\in\mathbb Z_q^m$, let $f(b,\bmx,\bm u)=A^T\bmx+\bm u+b\bm y$. Construct the state
		\begin{equation}\label{eq:state_LWE}
			\ket{\psi}_{AB}=\frac{1}{\sqrt{2q^n(2\eta+1)^m}}\sum_{b\in\{0,1\}}\sum_{\bmx\in\mathbb Z_q^n}\sum_{\bm u\in\mathbb [-\eta,\eta]^m}\ket{b,\bmx,\bm u}_A\ket{f(b,\bmx,\bm u)}_B
		\end{equation}
		and the observable
		\begin{equation}
			O=(\ketbra{0}{0}\otimes I_{n_A-1})\otimes(\ketbra{1}{1}\otimes I_{n_A-1})+(\ketbra{1}{1}\otimes I_{n_A-1})\otimes(\ketbra{0}{0}\otimes I_{n_A-1}).
		\end{equation}
		The state and the observable both have explicit polynomial-size quantum circuits to prepare or measure.
		
		We say a matrix $A$ is good, if for every non-zero $\bmx\in \mathbb Z_q^n$, $\|A^T\bmx\|_\infty>2\eta$. In this case, for every $b\in\{0,1\}$ and distinct pair $(\bmx,\bm u),(\bmx',\bm u')\in \mathbb Z_q^n\times [-\eta,\eta]^m$, $f(b,\bmx,\bm u)$ and $f(b,\bmx',\bm u')$ will be different.
		We next show that for all but exponentially small probability, $A$ is good. Fix an non-zero $\bmx\in \mathbb Z_q^n$. Suppose the $i$-th coordinate $x_i$ of $\bmx$ is non-zero, by the randomness of the $i$-th row of $A$, $A^T\bmx$ forms the uniform distribution over $\mathbb Z_q^m$. So
		\begin{equation}
			\Pr_{A}[\|A^T\bmx\|_\infty\leq 2\eta]= \left(\frac{4\eta+1}{q}\right)^m.
		\end{equation}
		Taking the union bound over all non-zero $\bmx$,
		\begin{equation}\begin{split}
				\Pr_{A}[\exists \bmx\in\mathbb Z_q^n,\bmx\neq \bm 0\land \|A^T\bmx\|_\infty\leq 2\eta]
				&\leq(q^n-1)\left(\frac{4\eta+1}{q}\right)^m\\
				&\leq(q^n-1)\left(\frac{16\beta n^{1.5}\log q+O(1)}{q}\right)^{2n\log q}\\
				&\leq(q^n-1)2^{-(2n-O(1))\log q}\\
				&\leq 2^{-(n-O(1))\log q}\\
				&=2^{-\Omega(n)}.
		\end{split}\end{equation}
		Here the second inequality follows from $\eta=\lceil 2\beta\sqrt nm+1\rceil$ and the third follows from $\beta\leq\frac{q}{32n^{1.5}\log q}$.
		
		Now we analysis the quantity $o^{(2)}$ for a good $A$. Suppose we measure the system $B$ and obtain outcome $\bmz$. Since $A$ is good, the post-measurement state $\psi_\bmz$ is either in form
		$\ket{b,\bmx,\bm u}$
		or in form
		$\frac{1}{\sqrt 2}(\ket{0,\bmx,\bm u}+\ket{1,\bmx',\bm u'})$.
		For the former case, $\tr(O\psi_\bmz^{\otimes 2})=0$; for the latter case, $\tr(O\psi_\bmz)^{\otimes 2}=\frac{1}{2}$. So
		\begin{equation}\begin{split}
				o^{(2)}_\psi&=\frac{1}{2}\Pr_{b,\bmx,\bm u}[\exists (\bmx',\bm u'),f(b,\bmx,\bm u)=f(1-b,\bmx',\bm u')]\\
				&=\frac{1}{2}\Pr_{\bmx,\bm u}[\exists (\bmx',\bm u'),f(0,\bmx,\bm u)=f(1,\bmx',\bm u')].
		\end{split}\end{equation}
		
		If $\bm y$ is drawn from the LWE distribution, we will have
		\begin{equation}\begin{split}
				f(1,\bmx,\bm u)&=A^T\bmx+\bm u+\bm y\\
				&=A^T(\bmx+\bms)+\bm u+\bm e\\
				&=f(0,\bmx+\bms,\bm u+\bm e).
		\end{split}\end{equation}
		Note that $\bmx+\bms\in \mathbb Z_q^n$ always holds. Moreover, if $\|\bm u+\bm e\|_\infty\leq \eta$ also holds, then choosing $\bmx'=\bmx+\bms,\bm u'=\bm u+\bm e$ yields $f(0,\bmx,\bm u)=f(1,\bmx',\bm u')$. Now we want to estimate $\|\bm e\|_\infty$ to show that $\|\bm u+\bm e\|_\infty\leq \eta$ with high probability.
		
		\begin{claim}
			\label{claim:gaussian}
			For every $t\in \mathbb N$,
			\begin{equation}
				\Pr_{\bm e\gets \cD_{\mathbb Z_q,\beta^2}^m}[\|\bm e\|_\infty>t]\leq \frac{2\beta^2 m}{t}e^{-t^2/2\beta^2}.
			\end{equation}
		\end{claim}
		
		Plug in $t=\beta\sqrt{n}$, then the probability that $\|\bm e\|_\infty>\beta\sqrt n$ is at most $\frac{2\beta m}{\sqrt n}e^{-n/2}$, which exponentially small in $n$. When $\|\bm e\|_\infty\leq \beta\sqrt{n}$, for every $\|\bm u\|_\infty\leq \eta-\beta\sqrt{n}$, we always have $\|\bm u+\bm e\|_\infty\leq \eta$. For such $\bm e$,
		\begin{equation}
			\begin{split}
				o^{(2)}_\psi&\geq \frac{1}{2}\Pr_{\bm u}[\|\bm u\|_\infty\leq \eta-\beta\sqrt{n}]\\
				&=\frac{1}{2}\left(\frac{2\eta-2\beta\sqrt n+1}{2\eta+1}\right)^m\\
				&\geq\frac{1}{2}\left(1-\frac{1}{2m}\right)^m\\
				&\geq\frac{1}{4}.
			\end{split}
		\end{equation}
		
		Now consider the case that $\bm y$ is drawn from the uniform distribution. Fix a pair of $(\bmx,\bm u)$. For every $\bmx'$,
		\begin{equation}\begin{split}
				\Pr_{\bm y\gets\cU(\mathbb Z_q^m)}[\exists \bm u',f(0,\bmx,\bm u)=f(1,\bmx',\bm u')]&=\Pr_{\bm y\gets \cU(\mathbb Z_q^m)}[\|f(0,\bmx,\bm u)-A^T\bmx'-\bm y\|_\infty \leq \eta]\\
				&=\Pr_{\bm y'\gets \cU(\mathbb Z_q^m)}[\|\bm y'\|_\infty \leq \eta]\\
				&=\left(\frac{2\eta+1}{q}\right)^m.
			\end{split}
		\end{equation}
		Taking a union bound over all $\bmx'$, we have
		\begin{equation}\begin{split}
				\EE{\bm y\gets \cU(\mathbb Z_q^m)}o^{(2)}_\psi&\leq\frac{1}{2}q^n\left(\frac{2\eta+1}{q}\right)^m\\
				&\leq 2^{-n\log q}.
		\end{split}\end{equation}
		By a Markov bound, for all but $2^{-(n\log q)/2}$ ratio of $\bm y$, we have $o^{(2)}_\psi\leq 2^{-(n\log q)/2}$.
		
		For $1-2^{-\Omega(n)}$ ratio of inputs, we can invoke $\cA$ and check whether the output is greater than $\frac{1}{8}$ to distinguish the two distributions with success probability $1-\delta$. Repeat this process for polynomially many times, we can boost the success probability to $1-2^{-\Omega(n)}$.
		
		The sizes $n_A$ and $n_B$ of the two systems should be at least $2m\log q$. Since $n_A,n_B\geq S(n_{AB})$, we can select $n_{AB}=S^{-1}(2m\log q)$ such that the state $\ket{\psi}_{AB}$ can be placed into the systems $A$ and $B$.
		Therefore, $\cA$ runs in $\poly(n_{AB})=\poly(S^{-1}(2m\log q))=\poly(S^{-1}(\cO(n\log^2 n)))$ time.
	\end{proof}
	\begin{proof}[Proof of Claim \ref{claim:gaussian}]
		For each coordinate $e_i$,
		\begin{equation}\begin{split}
				\Pr_{e_i\gets \cD_{\mathbb Z_q,\beta^2}}[|e_i|>t]&=\frac{2\sum_{x=t+1}^{q/2}e^{-x^2/2\beta^2}}{\sum_{z=-q/2}^{q/2}e^{-z^2/2\beta^2}}\\
				&\leq 2\sum_{x=t+1}^{q/2}e^{-x^2/2\beta^2}\\
				&\leq 2\int_{t}^{q/2}e^{-x^2/2\beta^2}\mathrm d x\\
				&\leq 2\int_{t}^{\infty}\frac{x}{t}e^{-x^2/2\beta^2}\mathrm d x\\
				&=\frac{2\beta^2}{t}e^{-t^2/2\beta^2}.
		\end{split}\end{equation}
		The first inequality holds because $e^{-z^2/2\beta^2}=1$ when $z=0$. The second inequality comes from the monotonicity of $e^{-x^2/2\beta^2}$.
		A union bound over all $m$ coordinates completes the proof.
	\end{proof}
	
	\begin{proof}[Proof of Theorem~\ref{thm:no-go-physical-prior}]
		Substituting $S(n_{AB})=\Omega(n_{AB}^c)$ into Lemma \ref{lemma:reduce-lwe}, the algorithm for DLWE runs in $\poly(n)$-time. Combining with Fact \ref{fact:desicion-to-search}, we complete the proof of Theorem \ref{thm:no-go-physical-prior}.
	\end{proof}
	
	\subsection{Proof of Corollary~\ref{corollary:entanglement}}
	
	For a pure state $\ket{\psi}_{AB}$, we consider the entanglement of post-measurement states on subsystem $A$.
	Specifically, we consider the bipartition $L|R$, where $L$ denotes the first qubit of $A$ and $R$ denotes the remaining qubits. We define the entanglement probability
	\begin{equation}
		e_\psi:=\Pr_{z\gets p_\psi}[\ket{\psi_z}_A\text{ is entangled across the cut }L|R]    
	\end{equation}
	as the probability that the post-measurement state exhibits entanglement across this partition.
	
	To establish the computational hardness of this task, it suffices to prove the following lemma, which demonstrates that estimating $e_\psi$ reduces to solving the DLWE problem. Recall that $n_{AB}$ denotes the total system size, while $n$ represents the size parameter for $\DLWE$.
	\begin{lemma}
		\label{lemma:entanglement-reduce-lwe}
		Let $0\leq\varepsilon< \tfrac{1}{10}$ and $0\leq\delta<\tfrac1{3}$ be fixed constants, and let $n_A,n_B\geq S(n_{AB})$ for some non-decreasing function $S$. Assume there exists an efficient algorithm $\cA$ such that given classical descriptions of a polynomial-sized quantum circuit $C$ for preparing a state $\ket\psi = C\ket{0}^{\otimes n_{AB}}$, it estimates $e_{\psi}$ such that
		\begin{equation}
			\abs{\mathcal{A}\bigl(C) - e_\psi} < \varepsilon
		\end{equation}
		with probability at least $1-\delta$, then there exists a $\poly(S^{-1}(\cO(n\log ^2 n)))$-time algorithm for $\DLWE_{n,m,q,\beta}$ with $m=2n\log q,q=\poly(n),\beta\leq\frac{q}{32n^{1.5}\log q}$ for prime $q$. Here $S^{-1}$ is the inverse function of $S$.
	\end{lemma}
	\begin{proof}
		Suppose we are given a matrix $A$ and a vector $\bm{y}$ as an instance of the DLWE problem. First, consider the case where $\bm{y}=0$. Since the probability of drawing the all-zero vector from the uniform distribution over $\mathbb{Z}_q^m$ is $2^{-\Omega(m)}$, we can conclude with high probability that $\bm{y}$ originates from the LWE distribution. We proceed assuming $\bm{y} \neq 0$.
		
		We construct the state $\ket{\psi}_{AB}$ following the same procedure as in the proof of Lemma~\ref{lemma:reduce-lwe}. 
		Let $L$ be the register for $b$ and $R$ be the register for $\bmx$ and $\bm u$ in Eq.~\eqref{eq:state_LWE}. We have shown that for a good $A$, each post-measurement state $\psi_\bmz$ is either in form $\ket{b}_L\ket{\bmx,\bm u}_R$ or in form $\frac{1}{\sqrt 2}(\ket{0}_L\ket{\bmx,\bm u}_R+\ket{1}_L\ket{\bmx',\bm u'}_R)$. The former state is separable. For the latter state, since $\bm y\neq 0$, the states $\ket{\bmx,\bm u}_R$ and $\ket{\bmx',\bm u'}_R$ in $R$ system cannot be the same. Hence, the state is entangled. Therefore, when $A$ is good and $\bm y$ is non-zero, $e_\psi=2o_\psi^{(2)}$ holds for the quantity $o_\psi^{(2)}$ defined in the proof of Lemma \ref{lemma:reduce-lwe}. One can use $\cA$ to estimate $o_\psi^{(2)}$ and hence solve DLWE.
	\end{proof}
	
	\begin{proof}[Proof of Corollary~\ref{corollary:entanglement}]
		Setting $S(n_{AB})=\log n_{AB}(\log\log n_{AB})^{2.1}\in \widetilde{\Theta}(\log n_{AB})$ gives $S^{-1}(x)=2^{\cO\left(\frac{x}{\log^{2.1} x}\right)}$. By Lemma \ref{lemma:entanglement-reduce-lwe}, we get an algorithm for DLWE runs in $\poly(S^{-1}(\cO(n\log ^2 n)))=2^{o(n)}$-time. Combining with Fact \ref{fact:desicion-to-search}, we complete the proof of Corollary~\ref{corollary:entanglement}.
	\end{proof}
\end{document}